\documentclass[sn-basic,Numbered]{sn-jnl}

\usepackage{graphicx}%
\usepackage{multirow}%
\usepackage{amsmath,amsthm,amssymb,amsfonts}%
\usepackage{mathrsfs}%
\usepackage[title]{appendix}%
\usepackage[pdftex,dvipsnames]{xcolor} 
\usepackage{textcomp}%
\usepackage{manyfoot}%
\usepackage{booktabs}%
\usepackage{algorithm}%
\usepackage{algorithmicx}%

\usepackage{listings}%
\usepackage{mathtools}
\usepackage{xspace}
\newtheorem{theorem}{Theorem}
\newtheorem{proposition}[theorem]{Proposition}%

\newtheorem{lemma}[theorem]{Lemma}%

\newtheorem{corollary}[theorem]{Corollary}%

\newtheorem{remark}{Remark}%

\newtheorem{notation}{Notation}

\newcommand{\cancel}[1]

\newcommand{\complex}{\mathsf{K}}

\newcommand{\middling}{p}

\usepackage{xr}
\usepackage{hyperref}
\usepackage{cleveref}

\makeatletter
\renewcommand\theHALG@line{\thealgorithm.\arabic{ALG@line}}
\makeatother

\usepackage{url}
\usepackage{graphicx}

\usepackage{csquotes}
\usepackage{nicefrac}
\usepackage{algorithm}
\usepackage[noend]{algpseudocode}
\algnewcommand{\LineComment}[1]{\State \(\triangleright\) #1}

\makeatletter
\newcommand{\multiline}[1]{%
  \begin{tabularx}{\dimexpr\linewidth-\ALG@thistlm}[t]{@{}X@{}}
    #1
  \end{tabularx}
}
\makeatother

\usepackage{mathtools}

\DeclarePairedDelimiter{\floor}{\lfloor}{\rfloor}
\usepackage{paralist}

\usepackage{csvsimple}
\usepackage{array}

\newcommand{\BCC}{\mathcal{B}}

\newcommand{\CCC}{\mathcal{C}}

\newcommand{\HCC}{\mathcal{H}}
\newcommand{\ECC}{\mathcal{E}}

\newcommand{\MCC}{\mathcal{M}}

\newcommand{\TCC}{\mathcal{T}}

\newcommand{\bbz}{\mathbb{Z}}

\newcommand{\sizeof}[1]{w(#1)}

\newcommand{\minHomCol}{\HCC}
\newcommand{\barminhom}{\overline{\HCC}}
\newcommand{\minCycleBasis}{\MCC}
\newcommand{\tightCycles}{\TCC}
\newcommand{\minBdryBasis}{\BCC}
\newcommand{\earliestbasis}[1]{\ECC(#1)}
\newcommand{\minTotalBasis}{\CCC}

\newcommand{\cycelem}{\zeta}
\newcommand{\bdryelem}{\eta}
\newcommand{\Hone}{\mathsf{H}_1}

\newcommand{\boldA}{\mathbf{A}}

\newcommand{\boldB}{\mathbf{B}}

\newcommand{\boldM}{\mathbf{M}}

\newcommand{\boldT}{\mathbf{T}}

\newcommand{\boldv}{\mathbf{v}}

\newcommand{\boldu}{\mathbf{u}}

\newcommand{\boldw}{\mathbf{w}}

\newcommand{\boldx}{\mathbf{x}}

\newcommand{\found}{\textnormal{found}}

\newcommand{\foundIndex}{\textnormal{foundIndex}}

\newcommand{\cp}{\mathsf{C}_{p}(K)}
\newcommand{\zp}{\mathsf{Z}_{p}(K)}
\newcommand{\bp}{\mathsf{B}_{p}(K)}
\newcommand{\hp}{\mathsf{H}_{p}(K)}

\newcommand{\hone}{\mathsf{H}_{1}(K)}
\newcommand{\hon}{\mathsf{H}_{1}}

\newcommand{\removed}[1]{}
\newcommand{\changed}[1]{#1}

\newcommand{\ztwo}{\mathbb{Z}_2}
\newcommand{\tightmatrix}[1]{\mathbf{T}(#1)}
\newcommand{\nnz}[1]{\mathsf{nnz}(#1)}
\newcommand{\newmatrix}{\mathbf{Y}}
\newcommand{\newmatrixtwo}{\mathbf{Z}}
\newcommand{\matrixit}[1]{\mathbf{#1}}

\newcommand{\fastloop}{FastLoop\xspace}

\DeclareMathOperator*{\im}{im}
\DeclareMathOperator*{\prob}{\textbf{Pr}}
\DeclareMathOperator*{\rk}{rk}

\usepackage{stackengine}
\setstackEOL{\\}
\makeatletter
\def\BState{\State\hskip-\ALG@thistlm}
\makeatother

\newcommand{\amritendu}[1] {{\sf\textcolor{orange}{{Amritendu: #1}}}}

\begin{document}

\title{Fast Algorithms for Minimum Homology Basis}

\author*[1]{\fnm{Amritendu} \sur{Dhar}}\email{amritendud@iisc.ac.in}
\equalcont{These authors contributed equally to this work.}

\author[1]{\fnm{Vijay} \sur{Natarajan}}\email{vijayn@iisc.ac.in}

\author[2]{\fnm{Abhishek} \sur{Rathod}}\email{arathod@post.bgu.ac.il}
\equalcont{These authors contributed equally to this work.}

\affil*[1]{\orgdiv{Department of Computer Science and Automation}, \orgname{Indian Institute of Science}, \orgaddress{ \city{Bangalore}, \postcode{560012}, \state{Karnataka}, \country{India}}}

\affil[2]{\orgdiv{Department of Computer Science}, \orgname{Ben Gurion University}, \orgaddress{ \city{Be'er Sheva}, \postcode{8410501},  \country{Israel}}}


\abstract{
We study the problem of finding a minimum homology basis, that is, a \changed{lightest} set of cycles that generates the $1$-dimensional homology classes with $\bbz_2$ coefficients in a given simplicial complex $K$. This problem has been extensively studied in the last few years. 
 \changed{For general complexes, the current best deterministic algorithm, by Dey et al.~\cite{DeyLatest}, runs in $O(N m^{\omega-1} + n m g)$ time, where $N$ denotes the total number of simplices in  $K$, $m$ denotes the number of edges in $K$, $n$ denotes the number of vertices in $K$,  $g$ denotes the rank of the $1$-homology group of $K$, and $\omega$ denotes the exponent of matrix multiplication.
In this paper, we present three conceptually simple randomized algorithms that compute a minimum homology basis of a general simplicial complex $K$. The first algorithm runs in $\tilde{O}(m^\omega)$ time, the second algorithm runs in  $O(N m^{\omega-1})$ time and the third algorithm runs in $\tilde{O}(N^2 g + N m g{^2} + m g{^3})$ time which is nearly quadratic time when $g=O(1)$}.

We also study the problem of finding a minimum cycle basis in an undirected graph $G$ with $n$ vertices and $m$ edges. The best known algorithm for this problem runs in $O(m^\omega)$ time. Our algorithm, which has a simpler high-level description, but is slightly more expensive, runs in $\tilde{O}(m^\omega)$ time.  

We also provide a practical implementation of computing the minimum homology basis for general weighted complexes. The implementation is broadly based on the algorithmic ideas described in this paper, differing in its use of practical subroutines. Of these subroutines, the more costly step makes use of a parallel implementation, thus potentially addressing the issue of scale. We compare results against the currently known state of the art implementation (ShortLoop).
}


\maketitle



\section{Introduction}

Minimum cycle bases in graphs have several applications, for instance, in
analysis of electrical networks, analysis of chemical and biological
pathways, periodic scheduling, surface reconstruction and graph drawing.
Also, algorithms from diverse application domains like electrical
circuit theory and structural engineering require cycle basis computation
as a preprocessing step. Cycle bases of small size offer a compact
description \changed{of representatives} that is advantageous from a mathematical  as well as
from an application viewpoint. For this reason, the problem of computing
a minimum cycle basis has received a lot of attention, both in its
general setting as well as in special classes of graphs such as planar
graphs, sparse graphs, dense graphs, network graphs, and so on. We
refer the reader to~\cite{KavithaSurvey} for a comprehensive survey.

In topological data analysis, \enquote{holes} of different dimensions in a geometric dataset
constitute \enquote{features} of the data. Algebraic topology offers
a rigorous language to formalize our intuitive picture of holes in
these geometric objects. More precisely, a basis for the first homology
group $\hon$ can be taken as a representative of the one-dimensional holes
in the geometric object. The advantages of using minimum homology
bases are twofold: firstly, one can bring geometry in picture by assigning
appropriate weights to edges, and secondly, smaller cycles are easier to
understand and analyze, especially visually. We focus solely on the bases
of the first homology group since the problem of computing a \changed{lightest} basis
 for higher homology groups with $\ztwo$ coefficients was shown to be NP-hard by Chen and Freedman~\cite{ChenFreedmanLocal}.

 \paragraph{Outline and Contributions}

In \Cref{sec:back}, we discuss the necessary preliminaries for cycle basis and homology basis computation. In \Cref{sec:mainone}, we describe a simple algorithm for computing \changed{ a } minimum cycle basis of a weighted graph.

In \Cref{sec:strucresult}, we prove a structural result relating minimum homology bases to minimum cycle bases. Specifically, we show that every minimum cycle basis \changed{of the} $1$-skeleton of a complex contains a minimum homology basis. In \Cref{sec:main}, we describe two randomized  algorithms \changed{(\Cref{alg:hombasis,alg:hombasistwo})}  for computing a minimum homology basis of a complex. In \Cref{sec:outsense}, we describe a third randomized algorithm \changed{(\Cref{alg:hombasisthree})} for the same problem that runs in nearly quadratic time when the first Betti number of the complex is a constant. All three algorithms use state-of-the-art black box matrix operations, and only the second  one (\Cref{alg:hombasistwo}) uses the structural result proved in \Cref{sec:strucresult}. \changed{\Cref{alg:hombasis} computes the column rank profile of a matrix consisting of tight cycles appended to the boundary matrix of the complex, whereas \Cref{alg:hombasistwo} computes the column rank profile of a matrix consisting of a minimum cycle basis of the $1$-skeleton of the complex appended to the boundary matrix of the complex.
\Cref{alg:hombasisthree} builds a  matrix $\boldB$ containing the minimum homology basis by iteratively finding a lexicographically smallest cycle from the set of tight cyles that is linearly independent of the current set of cycles stored in $\boldB$ by using a randomized binary search.}

To demonstrate that the ideas in this work have practical relevance, we provide an implementation of an algorithm based on  \Cref{alg:hombasistwo}. The implemented algorithm for computing minimum homology basis differs from  \Cref{alg:hombasistwo} in the use of matrix operations for computational efficiency. In particular, unlike \Cref{alg:hombasistwo}, it uses matrix reduction algorithm from the PHAT library~\cite{PHAT}. The details of the implementation can be found in \Cref{sec:implement}. Finally, in \Cref{sec:experiment}, we describe experiments on real-world datasets as well as random complexes. We show that our implementation \fastloop\footnote{\url{https://bitbucket.org/vgl_iisc/fastloop/}} consistently outperforms the state-of-the-art homology basis computation, ShortLoop~\cite{Shortloop, Shortloop_Paper}). 

Finally, we remark that this paper is the full and extended version of the conference version \changed{published in \cite{rathodminhom}.}

\section{Background and Preliminaries} \label{sec:back}

\subsection{Cycle Basis} \label{sec:cyclebasis}
Let $G = (V,E)$ be a connected graph. \changed{Throughout this paper, in context of graphs, we use $n$ to denote the number of vertices $|V|$ and $m$ to denote the number of edges $|E|$}. A subgraph of $G$ that has even degree for each vertex is called a \emph{cycle} of $G$. A cycle is called \emph{elementary} if the set of edges form a connected subgraph in which each vertex has degree $2$. We associate an incidence vector $C$, indexed on $E$, to each cycle, so that $C_e = 1$ if $e$ is an edge of the cycle, and $C_e = 0$ otherwise. The set of incidence vectors of cycles forms a vector space over $\ztwo$, called the \emph{cycle space} of $G$. It is a well-known fact that for a connected graph $G$, the cycle space is of dimension \changed{$m - n + 1$}. Throughout, we use $\nu$ to denote the dimension of the cycle space of a graph.
 A basis of the cycle space, that is, a maximal linearly independent set of cycles is called a \emph{cycle basis}. 

Suppose that the edges of $G$ have non-negative weights. Then, the weight of a cycle is the sum of the weights of its edges, and the weight of a cycle basis is the sum of the weights of the basis elements. The problem of computing a cycle basis of minimum weight \changed{(or lightest cycle basis)} is called the \emph{minimum cycle basis} problem. Since we assume all edge weights to be non-negative, there always exists a minimum cycle basis of elementary cycles, allowing us to focus on minimum cycle basis comprising entirely of elementary cycles.

\changed{Moreover, we define the \emph{length of a cycle} to be the number of edges in the cycle, and the \emph{length of a cycle basis} to be the sum of lengths of the cycles of the basis elements.}

\changed{An elementary cycle} $C$ is \emph{tight} if it contains a lightest path between every pair of points in $C$. We denote the set of all tight cycles in the graph by $\tightCycles$. Tight cycles are sometimes also referred to as \emph{isometric} cycles~\cite{AmaldiESA,KavithaSurvey}. Tight cycles play an important role in designing algorithms for minimum cycle bases, owing to the following theorem by Horton.

\begin{theorem}[Horton~\cite{Horton}] \label{thm:horton} A minimum cycle basis $\minCycleBasis$ consists only of tight cycles.
\end{theorem}

A key structural property about minimum cycle bases was proved by de Pina.

\begin{theorem}[de Pina~\cite{dePina}] \label{thm:depina}
Cycles $C_{1}\dots,C_{\nu}$ form a minimum cycle basis if \changed{there exists $m$-dimensional} vectors $S_{1},\dots,S_{\nu}$ \changed{that satisfy the following three conditions} for all $i$, $1\leq i\leq\nu$. 
\begin{description}
\item[Prefix Orthogonality:] $\langle \changed{C_{k}},S_{i}\rangle = 0$ for all $1\leq \changed{k}< i$. 

\item[Non-Orthogonality:] $\langle C_{i},S_{i}\rangle=1$.

\item[Shortness:] $C_{i}$ is a minimum weight cycle in $\mathcal{T}$ with $\langle C_{i},S_{i}\rangle=1$.

\end{description}
\end{theorem}

The vectors  $S_{1},\dots,S_{\nu}$ in~\Cref{thm:depina} are called \emph{support vectors}.
The recent line of algorithmic work~\cite{dePina,KavithaOld,KavithaNew,Mehlhorn,AmaldiESA} on the minimum cycle basis  problem rely on~\Cref{thm:depina}.
In fact, these algorithms may all be seen as refinements of the algorithm by de Pina, see \Cref{alg:depina}.

\begin{algorithm}[H]
\caption{De Pina's Algorithm for computing a minimum cycle basis}\label{alg:depina}
\begin{algorithmic}[1]
\vspace{1.5mm}

\State{Initialize $S_i$ to the $i$-th unit vector $e_i$ for $1 \leq i \leq \nu $ }
\For{$i \gets 1, \dots,\nu$}
	\State{Compute a minimum weight cycle $C_i$ with $\langle C_i, S_i \rangle = 1$. }
	\For{$j \gets i+1, \dots,\nu$}
		\State{$S_j = S_j +  \langle C_i, S_j\rangle S_i$ }
	\EndFor
\EndFor

\State{$\textsc{Return } \{ C_1, \dots, C_\nu \}$.}

\end{algorithmic}
\end{algorithm}

 \Cref{alg:depina} works by inductively maintaining a set of support vectors $\{S_i\}$ so that the conditions of~\Cref{thm:depina} are satisfied when the algorithm terminates.
In particular, Lines 4 and 5 of the algorithm ensure that the set
of vectors $S_{j}$ for $j>i$ are orthogonal to vectors $C_{1},\dots,C_{i}$.
Updating the vectors $S_j$ as outlined in Lines~4~and~5 of \Cref{alg:depina} takes time $O(m^{3})$ time
in total. Using a divide and conquer procedure for maintaining $S_j$, Kavitha et al.~\cite{KavithaOld} improved the
cost of maintaining the support vectors to $O(m^{\omega})$ \changed{effectively improving the cost of computing minimum cycle basis from $O(m^\omega n)$ to $O(m^2 n + m n^2 \log n)$, see \Cref{alg:kavitha}}.

\begin{algorithm}[H]
\caption{Divide and conquer procedure for fast computation of support vectors by Kavitha et al.~\cite{KavithaOld}}\label{alg:kavitha}
\begin{algorithmic}[1]
\vspace{1.5mm}

\State{Initialize $S_i$ to the $i$-th unit vector $e_i$ for $1 \leq i \leq \nu $.}
\State{MinCycleBasis$(1, \nu)$.}
\Statex
\Procedure{MinCycleBasis}{$\ell, u$}

	\If{$\ell = u$}
		\State{Compute a minimum weight cycle $C_\ell$ with $\langle C_\ell, S_\ell \rangle = 1$.}
	\Else
		\State{$q \gets \floor{\Large \frac{(\ell+u)}{2}}$.}
		\State{MinCycleBasis$(\ell, q)$.}
		\State{$\matrixit{C} \gets [C_\ell, \dots, C_q]$.}
		\State{$\matrixit{W} \gets(\matrixit{C}^{T}[S_{\ell},\dots,S_{q}])^{-1}\matrixit{C}^{T}[S_{q+1},\dots,S_{u}]$.}
		\State{$[S_{q+1},\dots,S_{u}]\gets[S_{q+1},\dots,S_{u}] + [S_{\ell},\dots.,S_{q}]\matrixit{W}$.}
		\State{MinCycleBasis$(q+1, u)$.}
	\EndIf

\EndProcedure


\State{$\textsc{Return } \{ C_1, \dots, C_\nu \}$.}

\end{algorithmic}
\end{algorithm}

\begin{lemma}[Lemma 5.6 in \cite{KavithaSurvey}]
The total number of arithmetic operations performed in lines 9 to 11 of \Cref{alg:kavitha} is $O(m^\omega)$. That is, the support vectors satisfying conditions of \Cref{thm:depina} can be maintained in $O(m^\omega)$ time.
\end{lemma}

Finally, in~\cite{AmaldiESA}, Amaldi et al. designed an $O(m^{\omega})$ time algorithm for computing a minimum cycle basis by improving the complexity of Line~5 of \Cref{alg:kavitha} to $o(m^\omega)$ (from $O(m^2 n)$ in~\cite{KavithaOld}), while using the $O(m^\omega)$ time divide-and-conquer template for maintaining the support vectors as presented in \Cref{alg:kavitha}.
The $o(m^\omega)$ procedure for \changed{Line~5} is achieved by performing a Monte Carlo binary search on the  set of tight cycles (sorted by weight) to find a minimum weight cycle $C_i$ that satisfies $\langle C_i, S_i \rangle = 1$.
An efficient binary search is made possible on account of the following key structural property about tight cycles.

\begin{theorem}[Amaldi et al.~\cite{AmaldiESA}] \label{thm:totallength}
\changed{The sum of lengths} of the tight cycles \changed{in a graph} is at most $n\nu$.
\end{theorem}

Amaldi et al.~\cite{AmaldiESA} also show that there exists an $O(nm)$ algorithm to compute the set of all the tight cycles of an undirected graph $G$. See Sections~2~and~3 of \cite{AmaldiESA} for details about Amaldi et al.'s algorithm.

\subsection{Matrix operations}

The \emph{column rank profile} (respectively \emph{row rank profile}) of an $m\times n$
matrix $\boldA $ with rank $r$, is the lexicographically smallest sequence
of $r$ indices $[i_{1},i_{2}, \dots ,i_{r}]$ (respectively $[j_{1},j_{2},\dots ,j_{r}]$)
of linearly independent columns (respectively rows) of $\boldA$.  Suppose that $\{a_1,a_2,\dots, a_n\}$ represent the columns of $\boldA$. Then, following Busaryev et al.~\cite{Busaryev}, we define the \emph{earliest basis} of $\boldA$ as the set of columns $\earliestbasis{\boldA} = \{a_{i_{1}},a_{i_{2}}, \dots ,a_{i_{r}}\}$. \changed{Throughout, we use $\nnz{\boldA}$ to denote the number of nonzero entries in matrix $\boldA$.}

It is well-known that classical Gaussian elimination can be used to compute rank profile in $O(nmr)$ time.
The current state-of-the-art deterministic matrix rank profile algorithms run in $O(m n r^{\omega-2})$ time.

\begin{theorem}[\cite{JEANNEROD201346,Dumas13}] \label{thm:deterprofile} There is a deterministic $O(m n r^{\omega-2})$ time algorithm to compute the column rank profile  of an  $m\times n$
matrix $\boldA$.
\end{theorem}

In case of randomized algorithms, Cheung, Kwok and Lau~\cite{Cheung} presented a breakthrough Monte Carlo algorithm for rank computation that runs in $(\nnz{\boldA}+r^{\omega})^{1+o(1)}$ time, where $o(1)$ in the exponent captures some missing multiplicative $\log n$ and $\log m$ factors. Equivalently, the complexity for Cheung et al.'s algorithm can also be written as $\tilde{O}(\nnz{\boldA}+r^{\omega})$. The notation $\tilde{O}(\cdot)$ is often used in literature to hide small poly logarithmic factors in time bounds.
While the algorithm by Cheung et al. also computes $r$ linearly independent columns of $\boldA$, the columns may not correspond to the column rank profile. However, building upon Cheung et al.'s work, Storjohann and Yang established the following result.

\begin{theorem}[Storjohann and Yang~\cite{Storjohann14,Storjohann15,Yang2014AlgorithmsFF}] \label{thm:storyang}
There exists a Monte Carlo algorithm for computing row (resp. column) rank profile of a matrix $\boldA$ that runs in $\changed{\tilde{O}(\nnz{\boldA} + r^\omega)}$ time. The failure probability of this algorithm is $\nicefrac{1}{2}$.
\end{theorem}

\removed{Once again, the $o(1)$ in the exponent captures some missing multiplicative $ \log n$ and $\log m$ factors, see~\cite{Storjohann14}, and hence the complexity can also be written as $\tilde{O}(\nnz{\boldA}+r^{\omega})$.}

In \Cref{sec:outsense}, we use Wiedemann's algorithms as subroutines to design an output sensitive algorithm to compute the minimum homology basis of a complex.

\begin{remark} \label{rem:wiedrem}
   Wiedemann~\cite{Wiedemann} presented randomized algorithms to compute the rank of an $m\times n$ matrix $\boldA$ over a finite field and for computing a solution to a sparse system of linear equations $\boldA x = b$ (if one exists). Both algorithms run in $\tilde{O}(n_1(\removed{\omega }\changed{\nnz{\boldA}}+ n_1))$ time, where $n_1 = \max(m,n)$ is the maximal dimension of the matrix $\boldA$\removed{ and $\omega$ is the total number of nonzero entries}. 
\end{remark}

\subsection{Homology}  
In this work, we restrict our attention to simplicial homology
with $\mathbb{Z}_{2}$ coefficients. For a general introduction to algebraic
topology, we refer the reader to \cite{MR1867354}.
Below we give a brief description of homology over $\mathbb{Z}_{2}$.

Let $K$ be a connected simplicial complex. We use $K^{(p)}$ to denote
the set of $p$-dimensional simplices of $K$, and $n_p$ the number of $p$-dimensional simplices of $K$. 
 Also, the $p$-dimensional
skeleton of $K$ is denoted by $K_{p}$. In particular, the $1$-skeleton
of $K$ (which is an undirected graph) is denoted by $K_{1}$.
\changed{Throughout this paper, in context of simplicial complexes, we use $n$ to denote $|K^{(0)}|$, $m$ to denote $|K^{(1)}|$, and $N$ to denote $|K|$.}

We consider formal sums of simplices with $\mathbb{Z}_{2}$ coefficients,
that is, sums of the form $\sum_{\sigma\in K^{(p)}}a_{\sigma}\sigma$,
where each $a_{\sigma}\in\{0,1\}$. The expression $\sum_{\sigma\in K^{(p)}}a_{\sigma}\sigma$
is called a $p$\emph{-chain}. Since chains can be added to each other,
they form an abelian group, denoted by $\mathsf{C}_{p}(K)$. Since
we consider formal sums with coefficients coming from $\mathbb{Z}_{2}$,
which is a field, $\mathsf{C}_{p}(K)$, in this case, is a vector
space of dimension $n_{p}$ over $\mathbb{Z}_{2}$. The $p$-simplices in $K$ give rise to the standard
basis for $\mathsf{C}_{p}(K)$.  This establishes a 
one-to-one correspondence between elements of $\mathsf{C}_{p}(K)$
and subsets of $K^{(p)}$. 
 Thus, associated with each chain is an incidence vector $v$, indexed on
$K^{(p)}$, where $v_\sigma = 1$ if $\sigma$ is a simplex of $v$, and $v_\sigma = 0$ otherwise.
 The \emph{boundary} of a $p$-simplex is a $(p-1)$-chain that corresponds
to the set of its $(p - 1)$-faces. This map can be linearly
extended from $p$-simplices to $p$-chains, where the boundary of
a chain is the $\mathbb{Z}_{2}$-sum of the boundaries of its elements.
Such an extension is known as the \emph{boundary homomorphism}, and
denoted by $\partial_{p}:\mathsf{C}_{p}(K)\to\mathsf{C}_{p-1}(K)$.
A chain $\zeta\in\mathsf{C}_{p}(K)$ is called a \emph{$p$-cycle} if $\partial_{p}\zeta=0$,
that is, $\zeta\in \ker \partial_{p}$. The group of $p$-dimensional
cycles is denoted by $\mathsf{Z}_{p}(K)$. As before, since we are
working with $\mathbb{Z}_{2}$ coefficients, $\mathsf{Z}_{p}(K)$
is a vector space over $\mathbb{Z}_{2}$. A chain $\eta\in\mathsf{C}_{p}(K)$
is said to be a $p$-boundary if $\eta=\partial_{p+1}c$ for some chain
$c\in\mathsf{C}_{p+1}(K)$, that is, $\eta \in \im \partial_{p+1}$. 
\changed{All $p$-boundaries are also $p$-cycles and sometimes referred to as \emph{trivial cycles} or \emph{null-homologous cycles}.}
The
group of p-\removed{dimensional }boundaries is denoted by $\mathsf{B}_{p}(K)$.
In our case, $\mathsf{B}_{p}(K)$ is also a vector space, and in fact
a subspace of $\mathsf{Z}_{p}(K)$.

 Thus, we can consider the quotient
space $\mathsf{H}_{p}(K)=\mathsf{Z}_{p}(K)/\mathsf{B}_{p}(K)$. The elements of the vector space $\mathsf{H}_{p}(K)$, known as the $p$-th
\emph{homology group} of $K$, are equivalence classes of $p$-cycles,
where $p$-cycles are equivalent if their $\mathbb{Z}_{2}$-difference is a $p$-boundary.
Equivalent cycles are said to be \emph{homologous}. For a $p$-cycle
$\zeta$, its corresponding homology class is denoted by $[\zeta]$.
Bases of $\mathsf{B}_{p}(K)$, $\mathsf{Z}_{p}(K)$ and $\mathsf{H}_{p}(K)$
are called \emph{boundary bases}, \emph{cycle bases}, and \emph{homology bases} respectively. The dimension
of the $p$-th homology group of $K$ is called the $p$-th \emph{Betti number
} of $K$, denoted by $\beta_p(K)$. We are primarily interested in the first Betti number $\beta_1(K)$. For notational convenience, let $g = \beta_1(K)$, and denote the dimension of $\mathsf{B}_{1}(K)$ by $b$.

Using the standard bases for $\mathsf{C}_{p}(K)$ and $\mathsf{C}_{p-1}(K)$,
the matrix $[\partial_{p}\sigma_{1}\:\partial_{p}\sigma_{2}\cdot\cdot\cdot\partial_{p}\sigma_{n_{p}}]$
whose column vectors are boundaries of $p$-simplices is called the
$p$-th boundary matrix. Abusing notation, we denote the $p$-th boundary
matrix by $\partial_{p}$. For the rest of the paper, we use $n, m$ and $N$ to denote the number of vertices, edges and simplices in the complex respectively.

A set of $p$-cycles $\left\{ \zeta_{1},\dots,\zeta_{g}\right\} $
is called a \emph{homology cycle basis} if the set of classes $\left\{ [\zeta_{1}],\dots,[\zeta_{g}]\right\} $
forms a homology basis. For brevity, we abuse notation by using the term \enquote{homology
basis} for $\left\{ \zeta_{1},\dots,\zeta_{g}\right\} $. 
Assigning non-negative weights to the edges of $K$,  the \emph{weight of a cycle} is the sum of the weights of its edges, and the \emph{weight of a homology basis} is the sum of the weights of the basis elements. The problem of computing a minimum weight  basis of $\hone$ \changed{(that is, a \emph{lightest basis} of $\hone$)} is called the \emph{minimum homology basis} problem. Note that, when the input simplicial complex is a graph, the notions of homology basis and cycle basis coincide. Please refer to~\Cref{fig:opgates} for an example.

\changed{Moreover, we define the \emph{length of a cycle} to be the number of edges in the cycle, and the \emph{length of a homology basis} to be the sum of lengths of the cycles of the basis elements.}

 \begin{figure}
\centering
\includegraphics[scale=0.72]{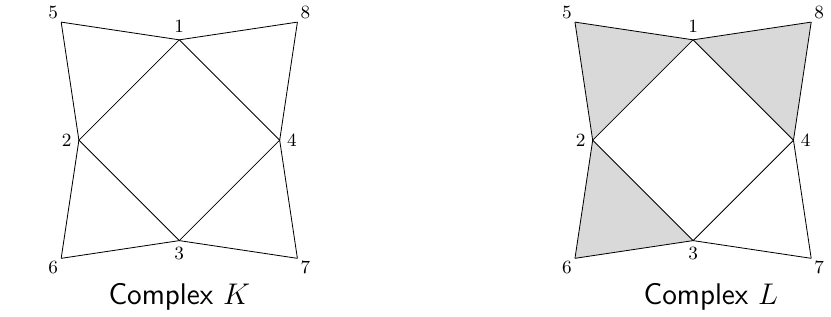} 

 \caption{Consider complexes $K$ and $L$ in the figure above with  unit weights on  the edges.  Since $K$ has no $2$-simplices, its $1$-skeleton $K_1$ is identical to $K$ itself.
The set of cycles $\mathcal{C} = \left\{  \{1,2,5\}, \{1,4,8\}, \{3,4,7\}, \{2,3,6\}, \{1,2,3,4\} \right \}$ constitutes a minimum cycle basis for the respective $1$-skeletons $K_1$ and $L_1$ (viewed as graphs). The set $\mathcal{C}$ also constitutes a minimum homology basis for $K$. The set $\mathcal{C}' =  \{\{1,2,3,4\}, \{3,4,7\} \}$ constitutes a minimum homology basis for  $L$.
   \label{fig:opgates}}
\end{figure}

For the special case when the input complex is a surface,	Erickson and Whittlesey~\cite{EW} gave a $O(n^{2}\log n +g n^{2}+g^{3} n)$-time algorithm. Recently, Borradaile et al.~\cite{Borradaile} gave an improved deterministic algorithm that runs in \changed{$O({g}^3 n \log n+m)$ assuming the lightest paths are unique}. For small values of \changed{$g$}, the algorithm in \cite{Borradaile} runs in nearly linear time. 

Furthermore, Dey et al.~\cite{Shortloop_Paper} and Chen and Freedman~\cite{ChenFreedman} generalized the results by Erickson and Whittlesey~\cite{EW} to arbitrary complexes. Subsequently, introducing the technique of annotations, Busaryev et al.~\cite{Busaryev} improved the complexity to $O(N^{\omega}+N^{2}g^{\omega-1})$. More recently, Dey et al.~\cite{DeyLatest} designed an $O(N^\omega + N^2 g)$ time algorithm by adapting the divide and conquer algorithm for computing a minimum cycle basis of Kavitha et al.~\cite{KavithaOld} for the purpose of computing a minimum homology basis. Dey et al. also designed a randomized $2$-approximation algorithm for the same problem that runs in $O(N^\omega \sqrt{N \log N})$ expected time. \changed{It is possible to etablish a tighter  bound for the algorithm by Dey et al.~\cite{DeyLatest} (See Section~\ref{sec:compare}).}

\subsection{Matroids} \label{sec:matroid}

A matroid $\mathscr{M}$ consists of a pair $(S,\text{\ensuremath{\mathscr{I}}})$,
where $S$ is a finite ground set and $\text{\ensuremath{\mathscr{I}}}$
is a family of subsets of $S$ satisfying the following axioms:
\begin{enumerate}
\item $\emptyset\in\mathscr{I}$;
\item if $I\in\mathscr{I}$ and $J\subseteq I$, then $J\in\mathscr{I}$;
and
\item if $I,K\in\mathscr{I}$ and $|I|<|K|$, then there is an element $e\in K\setminus I$
such that $I\bigcup\left\{ e\right\} \in\mathscr{I}$.
\end{enumerate}
If a set $I\subset S$ belongs to $\mathscr{I}$, then it is called
an \emph{independent set}; otherwise it is called a \emph{dependent set}.
A \emph{circuit} in a matroid $\mathscr{M}$ is a minimal dependent
subset of $S$. All proper subsets of  circuits are independent sets.
A maximal independent set is called a \emph{basis} of the matroid. 

Matroids admit a very useful property that goes by the name of \emph{basis
exchange property}: If $A$ and $B$ are distinct bases of a matroid
and $a\in A\setminus B$, then there exists an element $b\in B\setminus A$ such
that $A(\setminus\left\{ a\right\} )\bigcup\left\{ b\right\} $ is
again a basis. A \emph{weighted matroid} is a matroid $\mathscr{M}$
equipped with a weight function $w:S\to\mathbb{R}^{+}$ that additively
extends to all subsets of $S$. From an algorithmic standpoint, the
most important property of weighted matroids is that there is a greedy algorithm
for these matroids that computes the maximum (minimum) weight basis. It
can be easily checked that \textbf{$B$ }is a minimum weight basis
of a matroid if and only if none of the \changed{elements} of $B$ can be exchanged
for a \changed{lighter element} while still preserving linear independence. As
an immediate consequence, if the \changed{elements} in two distinct minimum \removed{cycle
(homology)} bases are sorted by weight, then the ordered sets of sorted
weights coincide. 

Let the cycle space of $G$ be the ground set \changed{$\Omega$}, and let $\mathscr{I}$ be defined as follows.

\[
\mathscr{I} = \{ I  \mid I \text{ is a linearly independent set of cycles of } G  \}
\]

Then, the pair $\mathscr{M}=(\changed{\Omega},\mathscr{I})$ forms a matroid. When
combined with a weight function on edges, it becomes a weighted matroid.
Cycle bases of $G$ correspond to the bases of $\mathscr{M}$. 

Analogously,  let the nontrivial $1$-cycles of a $2$-complex $\complex$ be the ground set \changed{$\Omega'$},
and let $\mathscr{I}'$ be defined as follows.

\[
\mathscr{I'} = \{ I  \mid I \text{ is a linearly independent set of nontrivial cycles of } \complex  \}
\]

Then, the pair $\mathscr{M}'=(\changed{\Omega'},\mathscr{I}')$ together with a weight function on the edges forms a weighted matroid. \changed{Sets of cycles whose classes form $1$-homology bases of $\complex$
are the bases of matroid $\mathscr{M}'$. }

Key to our algorithms is the property that the cycles in a cycle (homology) basis can be exchanged with other cycles while preserving independence.

\changed{
\begin{remark} \label{rem:opt}
    For any subset $\Omega_1$ of the ground set $\Omega$ (resp. any subset $\Omega_1'$ of the ground set $\Omega'$) with the property that $\Omega_1$ (resp. $\Omega_1'$) contains a minimum weight basis, the column rank profile of a matrix whose columns consist of cycles from $\Omega_1$ (resp. cycles from $\Omega_1'$) is \emph{the} greedy algorithm that returns a minimum weight basis.
\end{remark}
}

\section{An algorithm for computing \changed{ a }minimum cycle basis} \label{sec:mainone}

Given a graph $G=(V,E)$, let $ \{C_1, \dots, C_{|\tightCycles|} \}$ be the  list of tight cycles in $G$ sorted by weight, and let $\tightmatrix{G} =[C_1 \, C_2 \dots \, C_{|\tightCycles|}]$ be the matrix formed with cycles $C_i$ as its columns. Using~\Cref{thm:totallength}, since the total length of tight cycles is at most $n\nu$, and since each tight cycle consists of at least three edges, we have that $|\tightCycles| \leq \frac{n\nu}{3}$. Also, the rank of  $\tightmatrix{G}$ is $\nu$ and  $\tightmatrix{G}$ is a sparse matrix with $\nnz{\tightmatrix{G}}$ bounded by $n\nu$.  This sparsity is implicitly used in the design of the Monte Carlo binary search algorithm for computing \changed{ a } minimum cycle basis, as described in \cite{AmaldiESA}. We now present a simple and fast algorithm for minimum cycle basis that exploits the sparsity and the low rank of $\tightmatrix{G}$ more directly.

\begin{algorithm}[H]
\caption{Algorithm for minimum cycle basis}\label{alg:cyclebasis}
\begin{algorithmic}[1]
\vspace{1.5mm}

\State{Compute the sorted list of tight cycles in $G$, and assemble the matrix $\tightmatrix{G}$.}
\State{Compute the column rank profile $[i_{1},i_{2}, \dots ,i_{\nu}]$ of $\tightmatrix{G}$ using Storjohann and Yang's algorithm described in~\cite{Storjohann15}. }
\State{$\textsc{Return } \earliestbasis{\tightmatrix{G}}$.}

\end{algorithmic}
\end{algorithm}

\begin{theorem} \label{thm:algone} There is a Monte Carlo algorithm that computes the minimum cycle basis in  $\tilde{O}(m^\omega)$ time, with failure probability at most $\nicefrac{1}{2}$.
\end{theorem}
\begin{proof} The correctness of the algorithm follows immediately from \Cref{thm:horton}. 
As noted in \Cref{sec:matroid}, the cycles of a graph form a weighted matroid. \changed{Since the cycles in $\tightmatrix{G}$ are sorted by weight, and since the tight cycles of a graph are guaranteed to contain a minimum cycle basis, from \Cref{rem:opt} the column rank profile of $\tightmatrix{G}$ is the greedy algorithm, and the earliest basis $ \earliestbasis{\tightmatrix{G}}$ is a minimum cycle basis.}

The  list of tight cycles in $G$ can be computed in $O(n m)$ time using the algorithm described in Section~2 of \cite{AmaldiESA}. Hence, Step~1 of \Cref{alg:cyclebasis} takes $O(n m \log (n m))$ time (which in turn is same as $O(n m \log n )$ time).
Moreover, using \Cref{thm:storyang}, the complexity of  Step~2 is bounded by $\tilde{O}(n\nu + \nu^\omega)$. Since $n, \nu< m$, the complexity of \Cref{alg:cyclebasis} is bounded by $\tilde{O}(m^\omega)$. Using \Cref{thm:storyang}, the failure probability of the algorithm is at most $\nicefrac{1}{2}$.
\end{proof}

\section{Minimum homology basis, minimum cycle basis and tight cycles}
\label{sec:strucresult}

To begin with, note that since every graph is a $1$-dimensional simplicial complex, the minimum cycle basis problem is a restriction of the minimum homology basis problem to instances (simplicial complexes) that have no $2$-simplices. In this section, we refine this observation by deriving a closer relation between the two problems.


\changed{We will now define some notation that we will use in Lemma~\ref{lem:basecase} and \Cref{thm:keyresult}.}

\begin{notation} \label{not:bignote}
We assume that we are provided a complex $K$ in which all edges are assigned
non-negative weights.
\removed{Given a non-negative weight function $w: E \to \mathbb{R}^{+}$,
we define the weight of a cycle $z$ as the sum of the weights of
its edges, $w(z)=\sum_{\sigma\in z}w(\sigma)$.} 
Let $w: E \to \mathbb{R}^{+}$ be a non-negative weight function on the edges of a complex $K$, and let $\minBdryBasis = \{\bdryelem_1, \dots, \bdryelem_b\} $ be a basis for the boundary vector space $\mathsf{B}_{1}(K)$ indexed so that  $\sizeof{\bdryelem_i} \leq \sizeof{\bdryelem_{i+1}}$, $1\leq i < b$ (with ties broken arbitrarily).
Also, let $\minHomCol = \{\cycelem_1, \dots, \cycelem_g\}$ be a minimum homology basis of $K$ indexed so that $\sizeof{\cycelem_i} \leq \sizeof{\cycelem_{i+1}}$, $1\leq i <g$ (with ties broken arbitrarily).
Then, the set $\minTotalBasis = \{ \bdryelem_1, \dots, \bdryelem_b, \cycelem_1, \dots, \cycelem_g\}$ is a cycle basis for $K_1$.
Let $\minCycleBasis$ be a minimum cycle basis of $K_1$. Every element $C \in \minCycleBasis$ is homologous to a cycle $\sum_{i=1}^g a_i \cycelem_{i}$ where $a_i \in \{0,1\}$ for each $i$.
 Then,  for some fixed integers $p$ and $q$, $\minCycleBasis = \{B_1, \dots, B_q, C_1, \dots, C_p \}$ is indexed so that the elements $B_1, \dots, B_q$ are null-homologous and the elements $C_1,\dots, C_p $ are non-bounding cycles. Also, we have  $\sizeof{B_j} \leq \sizeof{B_{j+1}}$ for $1\leq j < q$ (with ties broken arbitrarily), and  $\sizeof{C_j}  \leq \sizeof{C_{j+1}}$ for $1\leq j < p$ (with ties broken arbitrarily).
\end{notation}

\changed{
 \begin{remark}
     We note that $p \geq  g$. This is because a null-homologous cycle can be obtained as a linear combination of non-bounding cycles, but a non-bounding cycle cannot be obtained as a linear combination of null-homologous cycles. Thus, there have to be at least $g$ linearly non-bounding cycles in $\minCycleBasis$ to be able to obtain all cycles in any homology basis from linear combinations of cycles in $\minCycleBasis$.
 \end{remark}
}
\begin{lemma} \label{lem:basecase} 
\changed{Given a simplicial complex $K$, suppose that $\mathcal{H}$, $\mathcal{C}$ and $\mathcal{M}$ are defined as in Notation~\ref{not:bignote}. The following two statements are true}
\begin{enumerate}
\item \changed{ $\sizeof{\cycelem_1} =  \sizeof{C_1}$}, and
\item  there exists a minimum homology basis $\barminhom$ with $\cycelem_1$  homologous to $C_1$. 
 \end{enumerate}
\end{lemma}
\begin{proof} 
Targeting \changed{a contradiction, suppose} there exists a minimum homology basis with  $\sizeof{\cycelem_1} < \sizeof{C_1}$. 
Let
 $\cycelem_1 = \sum_{i=1}^{p} a_i C_i  + \sum_{j=1}^{q} b_j B_j,$
where $a_i \in \{0,1\}$ for each $i$ and  $b_j \in \{0,1\}$ for each $j$. Since $\cycelem_1$ is a non-bounding cycle, there exists at least one $i$ with $a_i = 1$. Let $\ell \in [1,p]$ be the largest index in the above equation with $a_\ell = 1$. 
Rewriting the equation, we obtain
$ C_\ell =   \sum_{ i = 1}^{\ell - 1}  a_i C_i  + \sum_{j=1}^{q} b_j B_j  + \cycelem_1. $
Since $\sizeof{\cycelem_1} < \sizeof{C_1}$ by assumption, we have $\sizeof{\cycelem_1} < \sizeof{C_\ell}$ because $\sizeof{C_\ell} \geq \sizeof{C_1}$ by indexing of $\minCycleBasis$.
It follows that the basis obtained by exchanging $C_\ell$ for $\cycelem_1$, that is, $\{B_1, \dots, B_q, \cycelem_1, C_1,\dots, C_{\ell - 1},   C_{\ell + 1}, \dots, C_p\}$ gives a smaller cycle basis than the minimum one, a contradiction.

Once again, targeting a contradiction, suppose there exists a minimum homology basis with $\sizeof{\cycelem_1} > \sizeof{C_1} $.
Let 
$C_1 =  \sum_{i=1}^{g} a_i \cycelem_i  + \sum_{j=1}^{b} b_j \bdryelem_j $. As before, since $C_1$ is not null-homologous, there exists at least one $i$ with $a_i = 1$. Let $\ell \in [1,g]$ be the largest index in the above equation with $a_\ell = 1$. 
Then, $\cycelem_\ell =   \sum_{ i = 1}^{\ell - 1}  a_i \cycelem_i  + \sum_{j=1}^{b} b_j \bdryelem_j  + C_1.$
Note that $\sizeof{\cycelem_\ell} \geq \sizeof{\cycelem_1}$ because of the indexing, and $\sizeof{\cycelem_1} > \sizeof{C_1} $ by assumption.
 Therefore, the set $\{  C_1, \cycelem_1,\dots, \cycelem_{\ell - 1},  \cycelem_{\ell + 1}, \dots, \cycelem_p\}$  obtained by exchanging $\cycelem_\ell$ for $C_1$ gives a smaller homology basis than the minimum one, a contradiction.
 This proves the first part of the lemma.
   
  From the first part of the lemma, we have $\sizeof{\cycelem_1} = \sizeof{C_1}$ for every minimum homology basis. Let $\minHomCol$ be an arbitrary minimum homology basis. Then, if $C_1$ is not homologous to $\cycelem_1 \in \minHomCol$, by using basis exchange we can obtain $\barminhom = \{  C_1, \cycelem_1,\dots, \cycelem_{\ell - 1},  \cycelem_{\ell + 1}, \dots,$ $\cycelem_p\}$, which is the minimum homology basis with its first element  homologous to $C_1$, and having the same weight as $\sizeof{C_1}$, proving the claim.
\end{proof}

We now prove a theorem which allows us to harness fast algorithms for minimum cycle basis in service of improving time complexity of algorithms for minimum homology basis. 

\begin{theorem}\label{thm:keyresult} 
\changed{Given a simplicial complex $K$, suppose that $\mathcal{H}$, $\mathcal{C}$ and $\mathcal{M}$ are defined as in Notation~\ref{not:bignote}. Then,} there exists a minimum homology basis $\barminhom $ of $K$, and a set  $\{C_{i_1}, \dots, C_{i_g} \} \subset  \{C_1,\dots,  C_p\} \subset \minCycleBasis$  such that, for every $k \in [1,g]$, we have $C_{i_k}$ homologous to a cycle spanned by $\cycelem_1,\dots, \cycelem_k$, and $\sizeof{C_{i_k}} = \sizeof{\cycelem_k} $. Moreover, $i_1 = 1$, and $i_{k}$ for $k>1$ is the smallest index for which $C_{i_{k}} $  is not homologous to any cycle spanned by $\{C_{i_1}, \dots, C_{i_{k-1}} \}$.
  In particular,  the set $\{C_{i_1}, \dots, C_{i_g} \} \subset  \minCycleBasis$ constitutes a minimum homology basis of $K$.
\end{theorem}
\begin{proof}  
The key argument is essentially the same as for the proof of \Cref{lem:basecase}. Nonetheless, we present it here for the sake of completeness. We shall prove the claim by induction. \Cref{lem:basecase} covers the base case. By induction hypothesis, there is an integer \changed{$k < g$}, and a minimum homology basis $\minHomCol = \{\cycelem_1, \dots, \cycelem_g\}$, for which, vectors $\{C_{i_1}, \dots, C_{i_k} \} \subseteq  \{C_1,\dots,  C_p\}$ are such that, for every $j \in [1,k]$, we have $C_{i_j}$ homologous to a cycle spanned by $\cycelem_1,\dots, \cycelem_j$, and $\sizeof{C_{i_j}} = \sizeof{\cycelem_j} $. Let $i_{k+1}$ be the smallest index for which $C_{i_{k+1}} \in \{C_1,\dots,  C_p\}$  is not homologous to any cycle spanned by $\{C_{i_1}, \dots, C_{i_k} \}$. \changed{Such an index $i_{k+1} \leq p$ exists assuming every nontrivial cycle in $\mathcal{C}$ can be obtained as a linear combinaton of cycles in $\mathcal{M}$.} 

Suppose that $\sizeof{\cycelem_{k+1}} < \sizeof{C_{i_{k+1}}}$.
Let
 $\cycelem_{k+1} = \sum_{i=1}^{p} a_i C_i  + \sum_{j=1}^{q} b_j B_j$. 
Let $\ell \in [1,p]$ be the largest index in the above equation with $a_\ell = 1$. 
Then,
$ C_\ell =   \sum_{ i = 1}^{\ell - 1}  a_i C_i  + \sum_{j=1}^{q} b_j B_j  + \cycelem_{k+1}. $
From the induction hypothesis, we can infer that $\ell \geq i_{k+1}$, and hence $\sizeof{C_\ell} \geq \sizeof{C_{i_{k+1}}}$ by indexing of $\minCycleBasis$.
 Thus, if $\sizeof{\cycelem_{k+1}} < \sizeof{C_{i_{k+1}}}$, then we have $\sizeof{\cycelem_{k+1}} < \sizeof{C_\ell}$.
 It follows that, $\{B_1, \dots, B_q, \cycelem_{k+1}, C_1,\dots, C_{\ell - 1},  C_{\ell + 1}, \dots, C_p\}$ obtained by exchanging $C_\ell$ for $\cycelem_{k+1}$ gives a smaller cycle basis than the minimum one, contradicting the minimality of $\minHomCol$.

 Now, suppose that $\sizeof{\cycelem_{k+1}} > \sizeof{C_{i_{k+1}}} $. 
Let 
$C_{i_{k+1}} =  \sum_{i=1}^{g} a_i \cycelem_i  + \sum_{j=1}^{b} b_j \bdryelem_j $. Let $\ell \in [1,g]$ be the largest index in the above equation with $a_\ell = 1$. 
Rewriting the equation, we obtain $\cycelem_\ell =   \sum_{ i = 1}^{\ell - 1}  a_i \cycelem_i  + \sum_{j=1}^{b} b_j \bdryelem_j  + C_{i_{k+1}}.$
Again, using the induction hypothesis, $\ell \geq k+1$, and hence, $\sizeof{\cycelem_\ell} \geq \sizeof{\cycelem_{k+1}}$ because of the indexing. Since we have assumed $\sizeof{\cycelem_{k+1}} > \sizeof{C_{i_{k+1}}} $, this gives us $\sizeof{\cycelem_\ell} > \sizeof{C_{i_{k+1}}} $.
 Hence, the set $\{  C_{i_{k+1}}, \cycelem_1,\dots, \cycelem_{\ell - 1},  \cycelem_{\ell + 1}, \dots, \cycelem_p\}$ obtained by exchanging $\cycelem_\ell$ for $ C_{i_{k+1}}$ gives a smaller homology basis than the minimum one, contradicting the minimality of $\minHomCol$.
 
 From the  first part of the proof, we have established that  $\sizeof{C_{i_{k+1}}} = \sizeof{\cycelem_{k+1}} $.
 So, if $C_{i_{k+1}}$ is not homologous to $\cycelem_{k+1} \in \minHomCol$ and $\sizeof{\cycelem_{k+1}} = \sizeof{C_{i_{k+1}}}$, then  $ \barminhom = \{  C_{i_{k+1}}, \cycelem_1,\dots,$ $\cycelem_{\ell - 1},  \cycelem_{\ell + 1}, \dots, \cycelem_p\}$ obtained by exchanging $\cycelem_\ell$ for $ C_{i_{k+1}}$  is the desired minimum homology basis, proving the induction claim.
\end{proof}



Previously, it was known from Erickson and Whittlesey~\cite{EW} that $\minHomCol$ is contained in $\tightCycles$.

\begin{theorem} [Erickson and Whittlesey~\cite{EW}] \label{thm:EW} With non-negative weights, every cycle in a \changed{lightest} basis of $\Hone(K)$ is tight. 
That is, if $\minHomCol$ is any minimum homology basis of $K$, then $\minHomCol \subset \tightCycles$.
\end{theorem}

Using~\Cref{thm:keyresult,thm:horton}, we can refine the above observation.

\begin{corollary} \label{cor:maincor}
Let $\tightCycles$ denote the set of tight cycles of $K_1$, and let $\minCycleBasis$ be a minimum cycle basis of $K_1$. Then, there exists a minimum homology basis $\minHomCol$ of $K$ such that $\minHomCol \subset \minCycleBasis \subset \tightCycles$.
\end{corollary}


\section{Algorithms for minimum homology basis} \label{sec:main}

To begin with, note that since $\cp, \zp, \bp$ and $\hp$ are vector spaces, the problem of computing a minimum homology basis can be couched in terms of matrix operations.

Given  a complex $K$, let $ \{C_1, \dots, C_{|\tightCycles|} \}$ be the  list of tight cycles in $K_1$ sorted by weight, and let $\boldT =[C_1 \, C_2 \dots \, C_{|\tightCycles|}]$ be the matrix formed with cycles $C_i$ as its columns. 
Then, the matrix $\newmatrix = [\partial_2 \,\, | \,\, \boldT] $ has $O(N + n \nu)$ columns and $O(N+n \nu)$ non-zero entries since $\boldT$ has $O(n \nu)$ columns and $O(n \nu)$ non-zero entries by~\Cref{thm:totallength}, 
and $\partial_2$ has $O(N)$ columns and $O(N)$ non-zero entries.
Since $\newmatrix$ has $m$ rows, the rank of $\newmatrix$ is bounded by $m$. This immediately suggests an algorithm for computing minimum homology basis analogous to~\Cref{alg:cyclebasis}.

\begin{algorithm}[H]
\caption{Algorithm for minimum homology basis}\label{alg:hombasis}
\begin{algorithmic}[1]
\vspace{1.5mm}

\State{Compute the sorted list of tight cycles in $\boldT$, and assemble matrix $\newmatrix$.}
\State{Compute the column rank profile $[j_{1},j_{2}, \dots ,j_{b}, i_{1},i_{2}, \dots ,i_{g}]$ of $\newmatrix$ using Storjohann and Yang's algorithm~\cite{Storjohann15}, where columns $\{\newmatrix_{j_k}\}$ and $\{\newmatrix_{i_\ell}\}$ are linearly independent columns of $\partial_2$ and  $\boldT$ respectively.}
\State{$\textsc{Return } \textnormal{Columns  }\{\newmatrix_{i_1}, \newmatrix_{i_2}, \dots, \newmatrix_{i_g}\}$.}

\end{algorithmic}
\end{algorithm}

\begin{theorem} \label{thm:algtwo}  \Cref{alg:hombasis} is a Monte Carlo algorithm for computing a  minimum homology basis that runs in  $\tilde{O}(m^\omega)$ time with failure probability at most $\frac{1}{2}$.
\end{theorem}
\begin{proof} The correctness of the algorithm is an immediate consequence of Corollary~\ref{cor:maincor}. 
As noted in \Cref{sec:matroid}, the nontrivial cycles of a complex form a weighted matroid. \changed{Since the cycles in $\boldT$ are sorted by weight, and since the tight cycles of the $1$-skeleton of the complex is guaranteed to contain a minimum homology basis, from \Cref{rem:opt}, the column rank profile of $\boldT$ is the greedy algorithm that returns a minimum homology basis.}

The  list of tight cycles in $G$ can be computed in $O(n m)$ time using the algorithm described in Section~2 of \cite{AmaldiESA}. Hence, Step~1 of \Cref{alg:hombasis} takes $O(n m \log n)$ time.
Moreover, using \Cref{thm:storyang}, the complexity of  Step~2 is bounded by $\tilde{O}(N+n \nu + m^\omega)$, which is the same as $\tilde{O}( m^\omega)$ since $N$ and $n \nu$ are both in $\tilde{O}(m^\omega)$, and the failure probability is at most $\nicefrac{1}{2}$. 
\end{proof} 

When the number of $2$-simplices in complex $K$ is significantly smaller than the number of edges, the complexity for minimum homology can be slightly improved by decoupling the minimum homology basis computation from the minimum cycle basis computation, as illustrated in~\Cref{alg:hombasistwo}. 

\begin{algorithm}[H]
\caption{Algorithm for minimum homology basis}\label{alg:hombasistwo}
\begin{algorithmic}[1]
\vspace{1.5mm}

\State{Compute a minimum cycle basis $\minCycleBasis$ of  $K_1$ using the Monte Carlo algorithm by Amaldi et al.~\cite{AmaldiESA}. Let $\boldM$ be the matrix whose columns are cycle vectors in $\minCycleBasis$ sorted by weight. }
	\State{Assemble the matrix $\newmatrixtwo = [\partial_2 \,\, | \,\, \boldM] $. }
	\State{Compute the column rank profile $[j_{1},j_{2}, \dots ,j_{b}, i_{1},i_{2}, \dots ,i_{g}]$ of $\newmatrixtwo$ using the deterministic algorithm by Jeannerod et al.~\cite{JEANNEROD201346}, where columns $\{\newmatrixtwo_{j_k}\}$ and
	 $\{\newmatrixtwo_{i_\ell}\}$ are linearly independent columns of $\partial_2$
 and $\boldM$ respectively.}
	\State{$\textsc{Return } \textnormal{Columns  }\{\newmatrixtwo_{i_1}, \newmatrixtwo_{i_2}, \dots, \newmatrixtwo_{i_g}\}$.}

\end{algorithmic}
\end{algorithm}

\begin{theorem} \label{thm:algthree} \changed{A }minimum homology basis  can be computed  in \changed{$O(N m^{\omega-1})$} time using the Monte Carlo algorithm described in  \Cref{alg:hombasistwo}.  
The algorithm fails with  probability at most  $\nu\log(nm) \, 2^{-k}$, where  $k = m^{0.1}$.
\end{theorem}
\begin{proof} As in \Cref{thm:algtwo}, the correctness of the algorithm is an immediate consequence of Corollary~\ref{cor:maincor} \changed{and \Cref{rem:opt}}.
The algorithm fails only when Step~1 returns an incorrect answer, the probability of which is as low as $\nu\log(nm) \, 2^{-k}$, where $k = m^{0.1}$, see Theorem~3.2 of~\cite{AmaldiESA}.

The minimum cycle basis algorithm by Amaldi et al.~\cite{AmaldiESA} runs in $O(m^\omega)$ time (assuming the current exponent of matrix multiplication $\omega > 2$). 
Furthermore, using \Cref{thm:deterprofile}, the complexity of  Line~3 is bounded by $O(Nm^{\omega-1})$.
So, the overall complexity of the algorithm is $O(m^\omega + N m^{\omega-1})$ \changed{= $O(Nm^{\omega-1})$}. 
\end{proof}

Note that in Line~3 of \Cref{alg:hombasistwo},  it is possible to replace the deterministic algorithm by Jeannerod et al.~\cite{JEANNEROD201346} with the Monte Carlo algorithm by Storjohann and Yang's algorithm~\cite{Storjohann15}. In that case, the complexity of the algorithm will once again be $\tilde{O}(m^\omega)$, and the failure probability will be at most $1-  \frac{1}{2}(1-\nu\log(nm)2^{-k})$.


\section{A \changed{$g$-sensitive} algorithm for minimum homology basis} \label{sec:outsense}

In this section, we describe a randomized \changed{$g$-sensitive} algorithm for computing minimum homology basis whose complexity  depends on the value of $g$. Specifically, when $g=O(1)$, \Cref{alg:hombasisthree} computes the minimum homology basis of a complex correctly in nearly quadratic time.
To begin with, note that using Corollary~\ref{cor:maincor}, we know that the tight cycles of the $1$-skeleton $\complex_1$ of a complex $\complex$ contains a minimum homology basis of $\complex$. In this algorithm, the tight cycles of $\complex_1$ are maintained in a matrix denoted by $\boldT$. Specifically, the tight cycles are maintained in the columns of $\boldT$ and are sorted by weight.  Essentially, \Cref{alg:hombasisthree} builds a matrix $\boldB$ containing the minimum homology basis by iteratively finding a lexicographically smallest cycle in $\boldT$ that is linearly independent of the current set of cycles stored in $\boldB$.  It uses binary search each time to locate the lexicographically smallest linearly independent cycle. To make the search procedure efficient, randomization is used. Wiedemann's black box algorithms are used twice in \Cref{alg:hombasisthree}, first, in Line~\ref{lin:wiedemannone} to estimate $g$, and then in Line~\ref{lin:wiedemanntwo}  to check if a randomly selected cycle $\boldw$ is linearly independent of the basis cycles assembled in matrix $\boldB$.

Throughout this section, for some indices $i,j$, $\boldT[i]$ represents the $i$-th column of $\boldT$, and $\boldT[i\dots j]$ represents the submatrix of $\boldT$ formed by choosing columns $i$ through $j$.

\changed{Recall} that there are $g$ cycles in any homology basis.  In \Cref{alg:hombasisthree}, the \changed{outer} \textbf{for} loop in Lines~\ref{lin:bigforloopin}-\ref{lin:bigforloopout} runs $g'$ times, finding one linearly independent cycle in each iteration. Here, $g'=g$ with high probability since Wiedemann's algorithm computes the rank of a matrix correctly with high probability.
The \textbf{while} loop  in \Cref{alg:hombasisthree} uses a modification of binary search to find the lexicographically smallest cycles  that are linearly independent of cycles in the  matrix $\boldB$.
Suppose that we have some probabilistic guarantee that the first $k$ cycles in the minimum homology basis are correctly computed. Now, if $\boldT[\ell\dots \middling]$ has a cycle that is linearly independent of the cycles in $\boldB$, then in the \changed{probability amplification} for loop of Lines~\ref{lin:forloopin}-\ref{lin:forloopout} such a cycle is identified correctly with probability at least $\left(1 - \frac{1}{m^2}\right)$ (See Lemma~\ref{lem:msquare}).  On the other hand, if $\boldT[\ell\dots \middling]$ does not have a cycle that is linearly independent of the cycles in $\boldB$, then in the \textbf{if} condition of Line~\ref{lin:ifcond}, a  cycle that is linearly dependent on cycles in $\boldB$ is misidentified as linearly  independent with probability at most $\frac{1}{m^2}$ (See Lemma~\ref{lem:msquaretwo}). If a linearly independent cycle is successfully identified in $\boldT[\ell\dots \middling]$, then the search interval is halved by setting $r \gets  \left\lfloor \frac{\ell+r}{2}
\right\rfloor$ (See Line~\ref{lin:setr}). On the other hand, if the algorithm fails to find a linearly independent cycle in $\boldT[\ell\dots \middling]$, then in the next iteration of the \textbf{while} loop, the search interval is halved by setting $\ell \gets \left\lfloor \frac{\ell+r}{2}
\right\rfloor + 1$   (See Line~\ref{lin:setl}).

If we have narrowed down the search of the next  cycle in the basis to $\boldT[\ell]$ and  $\ell = r$ but $\boldT[\ell]$ is linearly dependent on cycles in $\boldB$, then the algorithm has clearly failed and is therefore terminated (See the \textbf{if} condition in Line~\ref{lin:iffail}). On the other hand, if we have narrowed down the search of the next  cycle in the basis to $\boldT[\ell]$ and  $\ell = r$ where $\boldT[\ell]$ is linearly independent of cycles in $\boldB$, then we add $\boldT[\ell]$ to $\boldB$ and initiate the search for the next cycle in the basis (See the \textbf{if} condition in Line~\ref{lin:ifsuccess}).





\cancel{
\begin{algorithm}[H]
\begin{algorithmic}[1]
\vspace{1.5mm}

\State{Compute a  minimum cycle basis of $\complex_1$ in matrix $\boldM$}



\State{Compute $g$ using Wiedemann's algorithm}
\amritendu{variable m denotes #edges and mid of binary search, ideally renamed.}
\State{$\ell \gets 1$; \quad $r\gets m-n+1$}

\For{$i = 1$ to $g$}
\While{$\ell \leq r$}
\State{$m\gets \left\lfloor \frac{\ell+r}{2}
\right\rfloor$; \quad $\found \gets 0$}
\For{$j = 1$ to $2 \log m$ } \label{lin:forloopin}
\State{Let $\boldv$ be a random $0$-$1$ vector of size $m-\ell +1$}
\State{$\boldw \gets \boldM[\ell \dots \middling] \cdot \boldv$}
\If{$\partial_2 \cdot \boldx = \boldw$ does not have a solution}
\If{ $i=1$ or ($i>1$ \quad \& \quad $\boldw ^T \boldB_{i-1}^{+} 
\boldB_{i-1}\neq \boldw^T)$}
\State{$r \gets m-1$; \quad $\found \gets 1$}
\State{$\foundIndex \gets m$}
\State{Exit the \textbf{for} loop in Lines~\ref{lin:forloopin}--\ref{lin:forloopout}}
\EndIf
\EndIf
\EndFor \label{lin:forloopout}
\If{$\found = 0$}
\State{$\ell \gets m + 1$}
\Else
\If{$r < \ell$}
\State{$\ell \gets \foundIndex + 1$}
\State{$r\gets m-n+1$}
\If{$i=1$}
\State{$\boldB_1 = \boldw^T$} 
\State{$(\boldB_1\boldB_1^T)^{-1} =(\boldw^T\boldw)^{-1} $}
\State{$\boldB_1^+ = \dfrac{\boldw}{\boldw^T\boldw}$ }
\Else
\State{$\boldB_i = \begin{pmatrix} \boldB_{i-1} \\
\boldw \end{pmatrix}$}
\State{$\boldu_i=\dfrac{\boldw^T(\boldB_{i-1}\boldB_{i-1}^T)^{-1}}{1+\boldw^T(\boldB_{i-1}\boldB_{i-1}^T)^{-1}\boldw}$}
\State{$(\boldB_i\boldB_i^T)^{-1}=(\boldB_{i-1}\boldB_{i-1}^T)^{-1}(\mathbf{I}-\boldw\boldu_i)$}
\State{$\boldB^+_{i} = \begin{pmatrix}\boldB^+_{i-1}(\mathbf{I}-\boldw\boldu_i)\\\boldu_i\end{pmatrix}$}
\EndIf
\EndIf
\EndIf
\EndWhile
\EndFor

\Statex

\Procedure{CheckBasisExtension}{$\ell, u$}
\EndProcedure

\caption{Algorithm for minimum homology basis}\label{alg:hombasistwo}
\end{algorithmic}
\end{algorithm}
}

\begin{algorithm}[h]
\begin{algorithmic}[1]
\vspace{1.5mm}

\State{Compute the set of tight cycles of $\complex_1$ in matrix $\boldT$ sorted by weight.} \label{lin:tightcycles}



\State{Use Wiedemann's algorithm to compute $\rk (\partial_1)$. Let $z_1 = m - \rk (\partial_1)$. Next, use Wiedemann's algorithm to compute $b_1 = \rk (\partial_2)$. Then, $g' \gets  z_1 - b_1$.} \label{lin:wiedemannone}
\LineComment{Wiedemann's algorithm computes rank correctly with high probability. Hence, $g'=g$ with high probability.}

\State{Initialize $\boldB$ with the empty matrix.}

\State{$\ell \gets 1$; \quad $r\gets m-n+1$;}

\For{$i = 1$ to $g'$} \Comment{\changed{Outer for loop}}\label{lin:bigforloopin}
\State{$k \gets 0$;} \Comment{The variable $k$ is not used in the algorithm; but is used in the analysis.}
\While{$\ell \leq r$} \Comment{The while loop is used for binary search.}
\label{lin:whilein}
\State{$\middling\gets \left\lfloor \frac{\ell+r}{2}
\right\rfloor$; \quad $\found \gets \changed{\texttt{false}}$;}
\For{\changed{$2 \log m$ times }} \Comment{\changed{Probability amplification for loop}}\label{lin:forloopin}
\State{Let $\boldv$ be a uniformly random $0$-$1$ vector of size $(\middling-\ell +1)$} 
\State{$\boldw \gets \boldT[\ell \dots \middling] \cdot \boldv$} \label{lin:selectrandom}
\State \Longunderstack[l]{ Use $2 \log m$ independent runs of Wiedemann's algorithm to \\ determine whether $[ \boldB \,\, | \,\, \partial_2 ]\cdot \boldx = \boldw$ has a solution} \label{lin:wiedemanntwo}
\If{ a solution to $[ \boldB \,\, | \,\, \partial_2 ]\cdot \boldx = \boldw$ is not found in any of the runs} \label{lin:ifcond}
 
\State{$r \gets \middling$; \quad $\found \gets \changed{\texttt{true}}$} \label{lin:setr}
\State{$\foundIndex \gets \middling$}
\State{Exit the \textbf{for} loop in Lines~\ref{lin:forloopin}--\ref{lin:forloopout}}
 
\EndIf
\EndFor \label{lin:forloopout}
\If{$\found = \changed{\texttt{false}}$} 

\If{$\ell=r$} \label{lin:iffail}
    \State{Print \enquote{\texttt{Algorithm failed.}}} \Comment{\changed{Binary search fails. Cycle not found.}}
    \State{\textbf{return};}
\Else
\State{$\ell \gets \middling + 1$} \label{lin:setl}
\EndIf
\Else \Comment{Binary search successful. Linearly independent cycle found.}
\If{$\ell = r $} \label{lin:ifsuccess}
\State{$\ell \gets \foundIndex + 1$}
\State{$r\gets m-n+1$}
\State{$\boldB \gets [\boldB \,\, | \,\,\boldw]$}
\State{Exit the \textbf{while} loop}
\EndIf
\EndIf
\State{$k \gets k+1$}
\EndWhile \label{lin:whileout}
\EndFor \label{lin:bigforloopout}

\caption{Algorithm for minimum homology basis}\label{alg:hombasisthree}
\end{algorithmic}
\end{algorithm}

\changed{
\begin{notation}
    Given an $m\times n$ matrix $\boldA$ and a $n$-dimensional column vector $\boldx$, the matrix vector product $\boldA \cdot \boldx$ to be equal to the vector $\sum_{i=1}^n \boldA_i \boldx_i$.
\end{notation}
}

\begin{lemma} \label{lem:half}
If there exists a column vector in $\boldT[\ell\dots \middling]$ that is linearly independent of column vectors of $[\boldB \mid \partial_2]$, then the probability that the vector $\boldw$  chosen in Line~\ref{lin:selectrandom} is such that the system of equations $ [ \boldB \,\, | \,\, \partial_2 ]\cdot \boldx = \boldw$ in Line~\ref{lin:wiedemanntwo}  does not have a solution is at least  $\frac{1}{2}$.    
\end{lemma}
\begin{proof} 
    We begin with noting that $\boldw$ is set to $\boldT[\ell \dots \middling] \cdot \boldv$ in \Cref{lin:selectrandom}.
    
    We will prove the claim in the lemma by induction.
    Let $i_1$ be the smallest index for which the column $\boldT[i_1]$ is not in the column space of $ [ \boldB \,\, | \,\, \partial_2 ]$. Then, setting $\boldv[i_1]$ to $1$ and $\boldv[j]$ to either $0$ or $1$ for $j\in \{ \ell,\dots,\changed{i_1 -1} \}$,  we obtain a set of linear \changed{combinations} of columns in $\boldT[\ell\dots i_1]$, denoted by $S_{i_1}$, which do not lie in the column space of $ [ \boldB \,\, | \,\, \partial_2 ]$. Since $\changed{|S_{i_1}| = 2^{i_1 - \ell }}$, half of the linear combinations of the first $i_1-\ell$ columns of $\boldT$ generate a column that does not belong to the column space of $[\boldB \mid \partial_2]$. This completes the base case of the induction.

  For the inductive hypothesis, assume that for some $i > i_1$, at least half of the \changed{$2^{i-\ell+1}$} linear combinations of columns of $\boldT[\ell\dots i]$ using the first \changed{$i-\ell + 1$} indices generate a column that is not in the column space of  $ [ \boldB \,\, | \,\, \partial_2 ]$. Denote this set of linear combinations by $S_i$. Denote the complementary set of linear combinations by $\overline{S}_i$. Said differently, $\changed{|S_i| \geq 2^{i- \ell }}$ and $\changed{|S_i \cup \overline{S}_i| = 2^{i-\ell+1}}$.
  
  Then, we have two cases: either $\boldT[i+1]$ is in the column space of $ [ \boldB \,\, | \,\, \partial_2 ]$ or not. 
  \begin{enumerate}[(1.)]
      \item  If it is in the column space of $ [ \boldB \,\, | \,\, \partial_2 ]$, then $S_{i+1}$ is obtained by extending the combinations in $S_{i}$ by setting $\boldv[i+1]$ to be either $0$ or $1$.
      \item  If it is not in the column space of $ [ \boldB \,\, | \,\, \partial_2 ]$, then the combinations in $S_{i+1}$ that generate a column that is not in the column space of $[\boldB \mid \partial_2]$ are obtained by 
       \begin{inparaenum}[(a.)]
      \item  extending combinations in $\overline{S}_i$ by setting $\boldv[i+1]=1$, 
      \item \changed{extending all combinations in $S_i$ by setting $\boldv[i+1] = 0$, and some combinations  in $S_i$ by setting $\boldv[i+1] = 1$ (if linear independence is preserved).} 
      \end{inparaenum}
  \end{enumerate}
  For every combination in $S_{i+1}$, the vector $\boldv[\ell\dots i+1]$ picks a column from $\boldT$ that is not in the column space of $ [ \boldB \,\, | \,\, \partial_2 ]$.
  Note that in both cases (1.) and (2.), \changed{$|S_{i+1}| \geq 2^{i-\ell+1}$} assuming \changed{$|S_i| \geq 2^{i-\ell}$}. The claim follows by noting that $\boldv$ is selected uniformly at random.
\end{proof}

\begin{lemma} \label{lem:msquare}
If there exists a column vector in $\boldT[\ell\dots \middling]$ that is linearly independent of column vectors of $[\boldB \mid \partial_2]$, then the \textbf{if} condition in Line~\ref{lin:ifcond} fails to satisfy in each of the $2 \log m$ iterations of the \changed{probability amplification} \textbf{for} loop (of Lines~\ref{lin:forloopin}-\ref{lin:forloopout}) with probability at most   $\frac{1}{m^2}$.   
\end{lemma}
\begin{proof}
    By Lemma~\ref{lem:half}, the failure probability of one iteration of the \changed{probability amplification} \textbf{for} loop is at most $\frac{1}{2}$. Since the \textbf{for} loop of Lines~\ref{lin:forloopin}-\ref{lin:forloopout} is executed $2 \log m$ times, and each time the vector $\boldw$ is chosen independently,  the \textbf{if} condition fails to satisfy with probability  at most $2^{-2 \log m} = \frac{1}{m^2}$.
\end{proof}

\begin{lemma} \label{lem:msquaretwo}
If there does not exist a column vector in $\boldT[\ell\dots \middling]$ that is linearly independent of column vectors of $[\boldB \mid \partial_2]$, then the \textbf{if} condition in Line~\ref{lin:ifcond} is satisfied with probability at most  $\frac{1}{m^2}$.   
\end{lemma}
\begin{proof}
By assumption, there does not exist a column vector in $\boldT[\ell\dots \middling]$ that is linearly independent of column vectors of $[\boldB \mid \partial_2]$. Suppose that in one of the iteration of the \changed{probability amplification} \textbf{for} loop (of Lines~\ref{lin:forloopin}-\ref{lin:forloopout})
all $2 \log m$ runs of the Wiedemann's algorithm fail to find a solution to  $[ \boldB \,\, | \,\, \partial_2 ]\cdot \boldx = \boldw$. The failure probability of one run of Wiedemann's algorithm (Line~\ref{lin:wiedemanntwo}) is at most $\frac{1}{2}$. Hence, the probability that $2 \log m$ runs fail to find a solution even when one exists is at most $2^{-2 \log m} = \frac{1}{m^2}$.
\end{proof}

 Let $\zeta_1$ be the lexicographically  smallest  nontrivial cycle in $\boldT$, and for every $i \in \{2,\dots,g\}$ let $\zeta_i$ be the lexicographically smallest cycle  in  $\boldT$ that is linearly independent of cycles $\zeta_j$ for $j\in [i-1]$. Let $\mathcal{E}_i$ be the event that $\boldB[i] = \zeta_i$. Also, let $\mathcal{E}_0$ be the event that $g' = g$.

 \begin{lemma} \label{lem:e1}
    $\prob[\mathcal{E}_1 
     \mid \mathcal{E}_0] \geq (1-\frac{1}{m^2})^{\left\lceil \log \nu\right\rceil}$.
\end{lemma}
\begin{proof}
For $i = 1$, in the \changed{outer} \textbf{for} loop of Lines~\ref{lin:bigforloopin}-\ref{lin:bigforloopout}, we denote the values of $\ell$, $\middling$ and $r$ in the $k$-th iteration of the \textbf{while} loop by $\ell_k$, $\middling_k$ and $r_k$, respectively.
Note that $\ell_0 = 1$ and $r_0 = m-n+1$. For $k \in \{2,\dots,\lceil \log \nu \rceil\}$, when $i=1$ and $k>1$, if the \textbf{if} condition \changed{in Line~\ref{lin:ifcond}} is successful, then  $\ell_k = \ell_{k-1}$ and $r_k = \middling_{k-1}$ (see Line~\ref{lin:setr}), else   $\ell_k = \middling_{k-1} + 1$ (see Line~\ref{lin:setl}) and $r_k = r_{k-1}$.
    The algorithm uses a modification of binary search to find the lexicographically smallest indexed column of $\boldT$ that does not lie in the column space of $[\boldB \mid \partial_2]$ (which is the same as $[\partial_2]$ when $i=1$).  

    Assuming $\mathcal{E}_0$ is satisfied, we have $g=g' \geq 1$.
    Clearly, $\zeta_1 \in \boldT[\ell_{0},r_{0}]$.
    Suppose that after $k-1$ iterations, the probability that $\zeta_1 \in \boldT[\ell_{k-1},r_{k-1}]$ is at least $(1-\frac{1}{m^2})^{k-1}$.
    If $\zeta_1 \in \boldT[\ell_{k-1},\middling_{k-1}]$, then using Lemma~\ref{lem:msquare}, $r_k = \middling_{k-1}$ and $\ell_k = \ell_{k-1}$ with probability at least $1 - \frac{1}{m^2}$. On the other hand, if $\zeta_1 \in \boldT[\middling_{k-1}+1,r_{k-1}]$, then  using Lemma~\ref{lem:msquaretwo}, $\ell_k = \middling_{k-1} + 1$ and $r_k = r_{k-1}$ with probability at least $1 - \frac{1}{m^2}$. In either case, the probability that $\zeta_1 \in \boldT[\ell_{k},r_{k}]$ is at least $(1-\frac{1}{m^2})^{k}$. 
    
    In every iteration of the \textbf{while} loop, the size of the search interval reduces by half. Since the total number of columns in $\boldT$ is $\nu$, 
    $\prob[\mathcal{E}_1 \mid \mathcal{E}_0]$ is at least $(1-\frac{1}{m^2})^{\left\lceil \log \nu\right\rceil}$.
\end{proof}   

\begin{lemma} \label{lem:ei}
    $\prob[\mathcal{E}_i \mid \cap_{j=0}^{i-1} \mathcal{E}_j] \geq (1-\frac{1}{m^2})^{\left\lceil \log \nu\right\rceil}$.
\end{lemma}
\begin{proof}
    The proof is nearly identical to Lemma~\ref{lem:e1}.
\end{proof}

We now recall a useful inequality from Motwani and Raghavan's book~\cite{motwani}.

\begin{proposition}[Proposition B.3 \cite{motwani}] \label{prop:mainprop}
For all $t,n\in\mathbb{R}$ such that $n\geq1$ and $t\le n$,
\[
e^{t}\left(1-\frac{t^{2}}{n}\right)\leq \left(1+\frac{t}{n}\right)^{n}.
\]
\end{proposition}

\begin{theorem} \label{thm:lastmain}
    \Cref{alg:hombasisthree} correctly computes the minimum homology basis with probability at least $\frac{1}{4} e^{-1}\left(1-\frac{1}{m^{2}}\right)$.
\end{theorem}

\begin{proof}
To begin with, note that Wiedemann's algorithm~\cite{Wiedemann} for computing the rank of a matrix has success probability at least $\frac{1}{2}$. Hence, $\prob[\mathcal{E}_0] \geq \frac{1}{4}$.

Using $\prob[\cap_{i=0}^{g'} \mathcal{E}_i] = \prob[\mathcal{E}_0] \times \prob[\mathcal{E}_1 \mid \mathcal{E}_0] \times \prob[\mathcal{E}_2\mid \mathcal{E}_1\cap \mathcal{E}_0]\times \dots \times \prob[\mathcal{E}_{g'} \mid \cap_{j=0}^{g'-1} \mathcal{E}_j]$ and Lemmas~\ref{lem:e1} and \ref{lem:ei}, we deduce that

\[\prob[\cap_{i=0}^{g'} \mathcal{E}_i] \geq \frac{1}{4}\left(1-\frac{1}{m^{2}}\right)^{\left\lceil \log\nu\right\rceil {g'}}. \]

    From  Proposition~\ref{prop:mainprop}, 
\[
e^{-1}\left(1-\frac{1}{m^{2}}\right)\leq\left(1+\frac{(-1)}{m^{2}}\right)^{m^{2}}\leq\left(1-\frac{1}{m^{2}}\right)^{\left\lceil \log\nu\right\rceil {g'}}
\]

Hence, the probability that  \Cref{alg:hombasisthree} correctly computes the minimum homology basis is given by $\frac{1}{4} \cdot e^{-1}\left(1-\frac{1}{m^{2}}\right)$. 
\end{proof}

\begin{theorem}
    \Cref{alg:hombasisthree} runs in $\tilde{O}(N^2 g + N m g{^2} + m g{^3})$ time.
\end{theorem}
\begin{proof}

The  list of tight cycles in $G$ can be computed in $O(n m)$ time using the algorithm described in Section~2 of \cite{AmaldiESA}. Hence, Line~1 of \Cref{alg:hombasisthree} takes $O(n m \log (n m)) = O(n m \log n )$ time.

Since $\partial_1$ and $\partial_2$ have $O(N)$ nonzero entries, rank computations, from \Cref{rem:wiedrem}, using Wiedemann's algorithm~\cite{Wiedemann} in Line~\ref{lin:wiedemannone} take $\tilde{O}(N^2)$ time.

We note that if $g' > g$, then in the \changed{outer} \textbf{for} loop of Lines~\ref{lin:bigforloopin}-\ref{lin:bigforloopout}, for $i=g+1$, the binary search will fail and the \textbf{if} condition in Line~\ref{lin:iffail} will be satisfied, and the algorithm will terminate. Therefore,  we can state the complexity analysis in terms of $g$ instead of $g'$.

   By \Cref{thm:totallength}, the total length of tight cycles of $\complex_1$ is at most $n \nu = O(nm)$. Using a sparse matrix representation of $\boldM$, Line~\ref{lin:selectrandom} takes $O(nm)$ time.
     From \Cref{rem:wiedrem}, we know that a single run of Wiedemann's algorithm takes $\tilde{O}((N+mg)\cdot(N+ g) ) = \tilde{O}(N^2 + N m g + m g{^2})$ time since $Ng = O(N^2)$. So, $\log m$ runs of Wiedemann's algorithm in Line~\ref{lin:wiedemanntwo} takes $\tilde{O}((N^2 + N m g + m g{^2}) 2 \log m) = \tilde{O}(N^2 + N m g + m g{^2})$ time. The \changed{outer} \textbf{for} loop in Lines~\ref{lin:bigforloopin}-\ref{lin:bigforloopout} runs at most $g+1$ times. The \textbf{while} loop in Lines~\ref{lin:whilein}-\ref{lin:whileout} runs at most $\log m$ times. The \changed{probability amplification} \textbf{for} loop in Lines~\ref{lin:forloopin}-\ref{lin:forloopout} runs $2 \log m$ times.
    Line~\ref{lin:wiedemanntwo} is executed at most $O(2 g \log^2 m)$ times.
    Since Line~\ref{lin:wiedemanntwo} is the most expensive step in the \changed{probability amplification} \textbf{for} loop from Lines~\ref{lin:forloopin}-\ref{lin:forloopout}, the complexity of the algorithm is $\tilde{O}(N^2 g + N m g{^2} + m g{^3})$.
\end{proof}

Note that when $g=O(1)$,  \Cref{thm:totallength} runs in nearly quadratic time.  

\section{\changed{Runtime comparison}} \label{sec:compare}
\changed{
Recall that \Cref{alg:hombasis} runs in $\tilde{O}(m^\omega)$,  \Cref{alg:hombasistwo} runs in $O(N m^{\omega-1})$ time and \Cref{alg:hombasisthree} runs in $\tilde{O}(N^2 g + N m g{^2} + m g{^3})$ time, where $n$ is the number of vertices, $m$ is the number of edges and $N$ is the total number of simplices.}

\changed{Note that Dey et al. prove a bound of  $O(N^{\omega} + N^{2}g)$ on the running time of their algorithm~\cite[Section 3.2]{DeyLatest}. However, a more refined analysis shows that the algorithm described in Dey et al. \cite{DeyLatest} runs in $O(nmg + Nm^{\omega-1})$ time. This is because the annotation algorithm  takes $O(Nm^{\omega-1})$ time in the worst case, whereas the \textsc{ShortestCycle} procedure in \cite{DeyLatest} takes $O(nm)$ time. Using the recurrence relation in \cite[Section 3.2]{DeyLatest} we obtain a time complexity bound of $O(Nm^{\omega-1} + nmg)$.}

\changed{ \paragraph*{Comparison of \Cref{alg:hombasis}~with~\Cref{alg:hombasistwo}} For families of complexes  with ${N}^{1-\epsilon} = \omega(m)$  for some $\epsilon >0$, \Cref{alg:hombasis} is faster than~\Cref{alg:hombasistwo}. However, for families of complexes such as triangulations of surfaces with  $N = \Theta(m)$, \Cref{alg:hombasistwo} is faster than \Cref{alg:hombasis}.} 

\changed{\paragraph*{Comparison of \Cref{alg:hombasis} and  \Cref{alg:hombasistwo} with Dey et al.'s algorithm}
For surfaces with $g = \Theta(m)$, \Cref{alg:hombasis,alg:hombasistwo} are faster than Dey et al.'s algorithm since Dey et al.'s algorithm takes at least $\Theta(m^{2.5})$ time since $ n = \Omega(m^{0.5})$, while \Cref{alg:hombasis,alg:hombasistwo} run in $\tilde{O}(m^\omega)$ and $O(m^\omega)$ time, respectively.} 

\changed{For dense simplicial complexes with $n$ vertices, $\Theta(n^2)$ edges and $\Theta(n^3)$ $2$-simplices, \Cref{alg:hombasis} is faster than Dey et al.'s algorithm since \Cref{alg:hombasis} runs in $\tilde{O}(n^{2\omega})$ time, whereas Dey et al.'s algorithm runs in $O(n^{2\omega+1})$ time.}

\changed{For surfaces with bounded $g$,
Dey et al.’s algorithm is slightly faster than \Cref{alg:hombasis} since Dey et al.’s algorithm runs in $O(m^\omega)$ time, whereas \Cref{alg:hombasis} runs in $\tilde{O}(m^\omega)$ time.}

\changed{Finally, it is easy to check that \Cref{alg:hombasistwo} is asymptotically always at least as fast Dey et al.’s algorithm, whereas in some important cases (such as surfaces with $g = \Theta(m)$ as discussed above), it is indeed much faster.
}

\changed{\paragraph*{Comparison of \Cref{alg:hombasisthree} with \Cref{alg:hombasis} and  \Cref{alg:hombasistwo}}
When $N = O(m)$ and $g$ is bounded, \Cref{alg:hombasisthree} is faster than \Cref{alg:hombasis,alg:hombasistwo}. On the other hand when $g = \Theta(m)$, \Cref{alg:hombasis,alg:hombasistwo} are faster than \Cref{alg:hombasisthree}. 
}

\changed{\paragraph*{Comparison of \Cref{alg:hombasisthree} with Dey et al.'s algorithm}
For surfaces with bounded $g$, \Cref{alg:hombasisthree} is faster than Dey et al's algorithm since \Cref{alg:hombasisthree} runs in $O(m^2)$ time whereas Dey et al's algorithm runs in $O(m^\omega)$ time. On the other hand, for surfaces with $g = \Theta(m)$,  Dey et al.'s algorithm is faster than \Cref{alg:hombasisthree} since \Cref{alg:hombasisthree} runs in $O(m^4)$ time whereas Dey et al.'s algorithm runs in $O(m^\omega)$ time. 
}

\cancel{

Given a matrix $\boldB$, we denote its pseudoinverse by $\boldB^{+}$. 

\noindent
{\bf Algorithm} {\sc OutSense} ($\mathsf{K}_\bullet$)\label{algorithm:cuppers}
\begin{itemize}
    \item Step 1. Compute the set of tight cycles of $\complex_1$ in matrix $\boldT$.
     \item Step 2. Using Wiedemann's algorithm compute $z_1 = m - \rk (\partial_1)$. Next, using Wiedemann's algorithm compute the $b_1 = \rk (\partial_2)$.
     \item Step 3.  Then, $g \gets  z_1 - b_1$.
    \item Step 4.  Initialize \quad $\ell \gets 1$; \quad $r\gets m-n+1$.

\end{itemize}

\begin{lemma} {\bf Algorithm} {\sc OutSense} runs in $\tilde{O}(m^2 g)$ time.
\end{lemma}

If $\boldw^T \in \mathcal{R}(\boldB)$, then $\exists \boldx^T$ such that
\begin{align*}
\boldw^T &= \boldx^T \boldB    \\
         &= \boldx^T \boldB \boldB^+ \boldB \\
         &= \boldw^T \boldB^+ \boldB
\end{align*}
On the other hand, if $\boldw^T = \boldw^T \boldB^+ \boldB$, then clearly $\boldw^T \in \mathcal{R}(\boldB)$.
Moreover, $\boldw^T = \boldx^T \boldB$ has a solution if and only if $\boldw^T \in \mathcal{R}(\boldB)$. This justifies line 11.
}

\section{Implementation} \label{sec:implement}
While the algorithms described in the previous section have good worst case runtime bounds, they are not amenable to a simple and efficient implementation. \changed{We use  ideas from \Cref{alg:hombasistwo} to implement an algorithm that exhibits good performance in practice by leveraging existing state-of-the-art software for matrix reduction and minimum cycle basis computation.} We begin by describing the data representations for input/output and intermediate storage.

\subsection{Input/output format}\label{sec:ip-op-format}
Geometric simplicial complex input are expected to be in the \texttt{OFF} \cite{OFFFormat} file format. The current implementation assumes that the complex is embedded in $\mathbb{R}^{3}$ but this may be extended to high dimensional Euclidean spaces. The weight of an edge in the complex is assumed to be given by the Euclidean distance between its end point vertices.
A general simplicial complex input may be specified in a simple text file that stores the 2-skeleton of the complex. The text file contains the number of vertices, edges, and triangles, followed by a list of weighted edges and finally a list of triangles of the complex, all in ASCII format. 
An edge is represented as a 3-tuple $\{ i,j,w\}$ in a separate line where $i,j$ are the indices of its end point vertices and $w$ is the edge weight. A triangle is represented as a 3-tuple $\{ i,j,k\}$ in a separate line where $i,j,k$ are the indices of the vertices of the triangle. 

The output is available in a single text file \changed{consisting of the betti} number ($\beta_1$) followed by a list of cycles that represent a minimum homology basis. Each cycle is represented as a sequence of vertex indices.

\subsection{Internal data representation}
We maintain the following in-memory data structures for querying the input simplicial complex and storing the results of the intermediate steps of the algorithm.
\begin{enumerate}
    \item \emph{vertexList:} A list of vertices and their location in $\mathbb{R}^3$ for geometric simplicial complexes. 
    \item \emph{triangles:} A list of triangles of the complex. Each triangle is stored as a 3-tuple representing the index of its three vertices in \emph{vertexList}.
    \item \emph{vPairToE:} Edges of the complex are enumerated. The index \emph{edgeNo} of an edge ranges from $1$ to $m$. \emph{vPairToE} is a map where each key is a pair of vertex indices corresponding to the end points of an edge and the associated value is the \emph{edgeNo} of that edge.
    \item\emph{eToVPair:} The inverse map of \emph{vPairToE}. A map that specifies the pair of vertex indices corresponding to an edge. 
    \item \emph{eToWeightMap:} A map from an edge index to the edge weight.
    \item \emph{graph:} The 1-skeleton of the input simplicial complex, stored as a Boost adjacency list~\cite{boost}.
\end{enumerate}

\subsection{Algorithm}
We now describe \fastloop, a practical implementation of \Cref{alg:hombasis}.  \fastloop employs alternate algorithms for computing the minimum cycle basis and column rank profile. These algorithms are efficient in practice, amenable to parallel computation, and their implementation is made available within reliable software libraries. Below, we provide an overview of the main steps of \fastloop. 
\begin{enumerate}
    \item  Load the input simplicial complex and populate the key data structures described in the previous section.
    \item Compute a minimum cycle basis (see \Cref{Compute-MCB}).
    \item Assemble the matrix $\newmatrix$, the boundary matrix prepended to the minimum cycle basis (\Cref{alg:hombasistwo})
    \item Employ a column reduction algorithm that adds columns from left to right to compute the column rank profile (see \Cref{compute-MHB}).
\end{enumerate}

\subsubsection{Computing a minimum cycle basis}\label{Compute-MCB}
We use parmcb, the Parallel Minimum Cycle Basis library~\cite{parmcb}, to compute a minimum cycle basis of the 1-skeleton of the input simplicial complex. The parmcb library implements a suite of algorithms that are broadly based on De Pina's algorithm~\Cref{alg:depina} but differ in the step that computes a minimum weight cycle that is required in Step 3 of the algorithm.
It supports parallel execution using MPI and Intel TBB. The input graph is represented using the graph data structure from the Boost library. For our purpose, we build a Boost library adjacency list representation of the 1-skeleton of the 2-complex~(\emph{graph}). The minimum cycle basis is reported as a collection of cycles, where each cycle is represented as a list of edges.

\subsubsection{Computing the minimum homology basis}\label{compute-MHB}
Any matrix reduction algorithm that adds columns from left to right can be used to compute the column rank profile. Specifically, the indices of the non-zero columns at the end of such a  reduction procedure gives the column rank profile. We use the standard reduction algorithm in the PHAT library~\cite{PHAT}. PHAT is a library of matrix reduction algorithms~\cite{Bauer2017Phat} implemented in C++ for computing barcodes in persistent homology. We recall the standard reduction algorithm in \Cref{alg:standardReduction}. \Cref{thm:totallength} implies that the matrix $\newmatrix$ is sparse. This motivates the use of PHAT, which is optimized to exploit the sparsity of the underlying matrix. PHAT supports multiple sparse representations of the matrix. We choose the bit\_tree representation~\cite[Section 4]{Bauer2017Phat}, where each column is represented as a balanced binary tree of row indices that contain a 1 in that column. The matrix is maintained as a list (\emph{vector}) of such balanced binary trees. 

\begin{algorithm}[h]
\caption{Standard reduction algorithm for matrix reduction~\cite{Bauer2017Phat}}\label{alg:standardReduction}
\begin{algorithmic}[1]
\vspace{1.5mm}

\State{Input: A 0-1 matrix $\partial$ with $m$ columns.}
\State{$low(j)$ denotes the row index of the lowest 1 in $\partial$, it is undefined if column $j$ contains only zeros}.
\For{$j = 1$ to $m$ }
\While{there exists $j_0 < j$ with $low(j_0) = low(j)$}
\State{Add column $j_0$ to $j$}
\EndWhile
\EndFor
\end{algorithmic}
\end{algorithm}

\section{Experimental results} \label{sec:experiment}
We now describe results from our computational experiments on various real world and synthetic datasets. The experiments serve two primary purposes. First, they validate the correctness of \fastloop. Second, they reveal the efficiency of the algorithm when measured against the size and type of the input complex and the number of CPU cores deployed. \changed{All experiments were performed on an Intel workstation powered by a Xeon(R) Gold 6230 CPU with 20 cores at 2.10 GHz and 384~GB RAM running Ubuntu Linux}. Parallelization in computing the minimum cycle basis was achieved using Intel Thread Building Blocks~(TBB).

\subsection{Cycle representatives}
The real world datasets include 2D meshes and 3D volume meshes, both manifold and non-manifold. \Cref{table: execution_time} lists all datasets and \Cref{fig:Mhb_examples} shows results on a subset of the datasets. The top two rows in \Cref{fig:Mhb_examples} are polygonal meshes from the Visionair shape repository~\cite{visionair}. These meshes are medium sized datasets consisting of 30000--70000 simplices. 
The accompanying video (Online Resource~1) shows the computed cycle representatives from different view points.

Computational experiments on the real world datasets help demonstrate the practical utility of the algorithm via direct visualization of the loops and tunnels. Our algorithm computes optimal cycle representatives for holes of different sizes. The third and fourth rows in \Cref{fig:Mhb_examples} show data from the PCOD hypothetical zeolite database~\cite{Zeolitedataset}. Zeolite structures are known to contain pores.  Hypothetical structures are generated computationally and their properties are studied to determine if they are similar to existing zeolites with desirable properties. A distance field that captures distance from a point to its nearest atom is computed for each hypothetical zeolite structure using the Zeo++ software~\cite{ZeolitePaper}. The structure of three hypothetical zeolite materials are visualized by rendering the zero isosurface, the preimage of distance value 0.
Zeolites are known to contain pores. Identifying the pores and quantifying their size is an important problem because the pore size determines the chemical properties of the zeolite. Zeolites also have a spatially repeating structure that contributes to a larger value of $\beta_1$. 

\Cref{fig:Mhb_Protein} shows the surface of two protein molecules, 3EAM and 1OED. Both proteins contain a central tunnel and multiple long pores. A subset of the cycles of the minimum 1-homology basis, including the central tunnel, are highlighted in middle column. All the cycles of the  minimum 1-homology basis are highlighted in the right column.

\subsection{Verifying correctness}
Notwithstanding the theoretical correctness of the algorithm, we undertake several measures to ensure correctness of the implementation in \fastloop. The two major components of \fastloop are based on highly optimized, well maintained, stable, and well tested software. Specifically, the library parmcb~\cite{parmcb} is used for computing the minimum cycle basis~(MCB), and PHAT~\cite{PHAT} is used for the reduction step to compute the minimum homology basis~(MHB) from the MCB. We also compare the output of \fastloop and ShortLoop~\cite{Shortloop_Paper,Shortloop} on the real world datasets. The Betti numbers reported by both software match each other, and for a majority of the datasets the MHB reported by both software also match. However, we found a few examples where the weight of the MHB computed by \fastloop  differed from that reported by ShortLoop. In order to explain this discrepancy, we performed a few sanity checks on the outputs of Shortloop and \fastloop. In particular, we checked if the loops reported by the software are indeed non-bounding and independent. We observed that a few cycles reported by ShortLoop fail the independence check, specifically in the cases when ShortLoop reports a smaller weight basis when compared to \fastloop. This discrepancy can also be detected visually in some instances. \Cref{fig:Mhb_discrepancy} shows an example where the 
\changed{ collection of basis cycles}  reported by ShortLoop (red) misses the central hole of the wheel. The developers of ShortLoop have been informed about this issue.

\begin{figure*}[!htb]
\centering
\begin{tabular}{cc}
\includegraphics[height=1.7in]{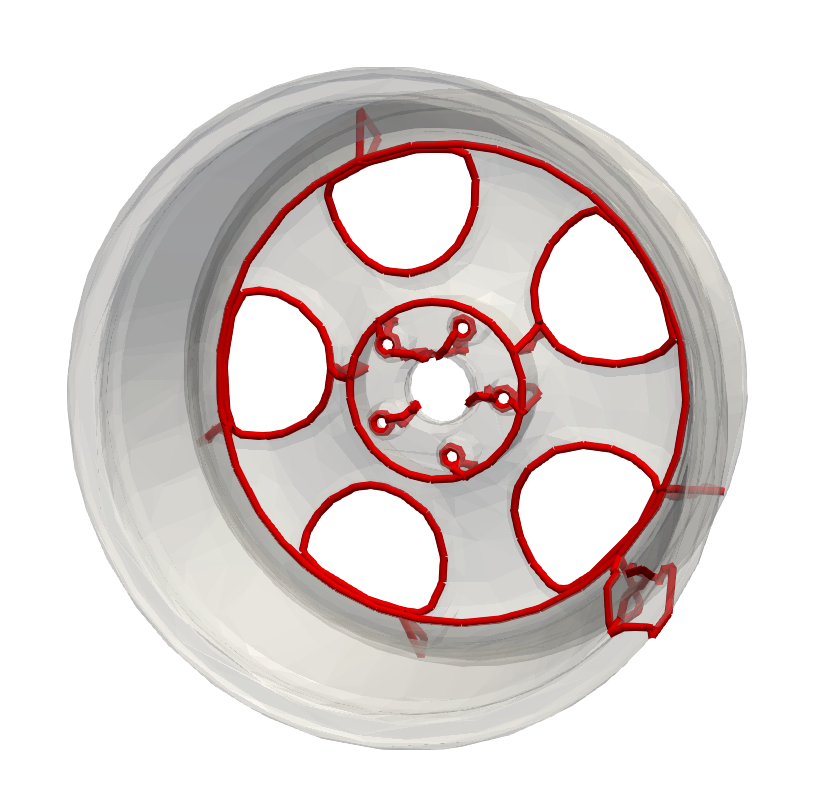} &
\includegraphics[height=1.7in]{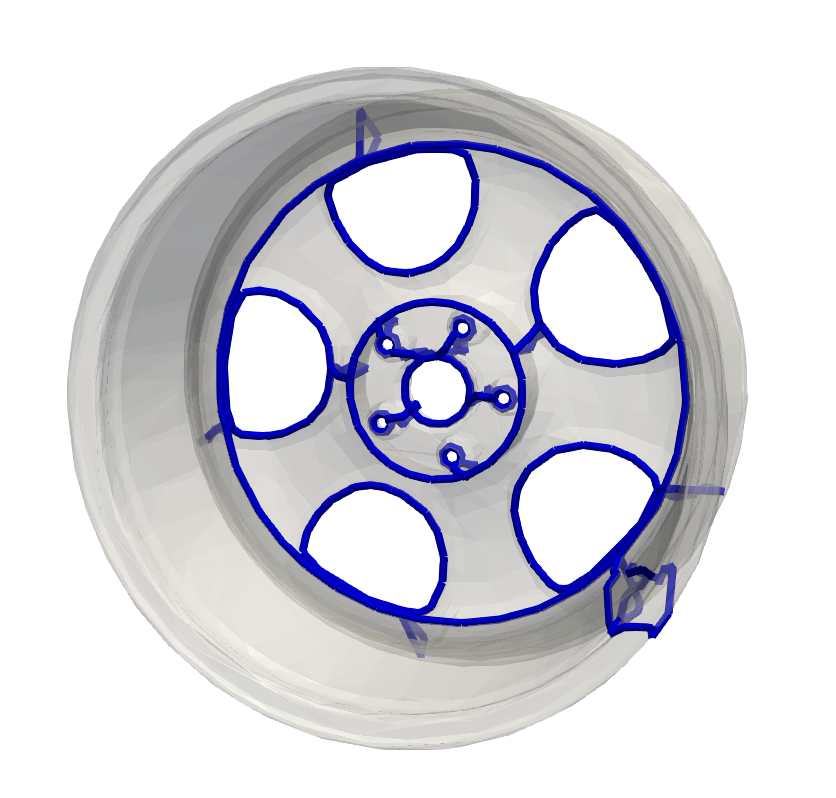} \\
ShortLoop & \fastloop \\
\end{tabular}
\caption{ShortLoop misses the central hole (left) whereas \fastloop correctly identifies the cycle representing the hole.}
    \label{fig:Mhb_discrepancy}
\end{figure*}

\begin{figure*}[!htb]
\centering
\begin{tabular}{ccc}
\includegraphics[height=1.55in]{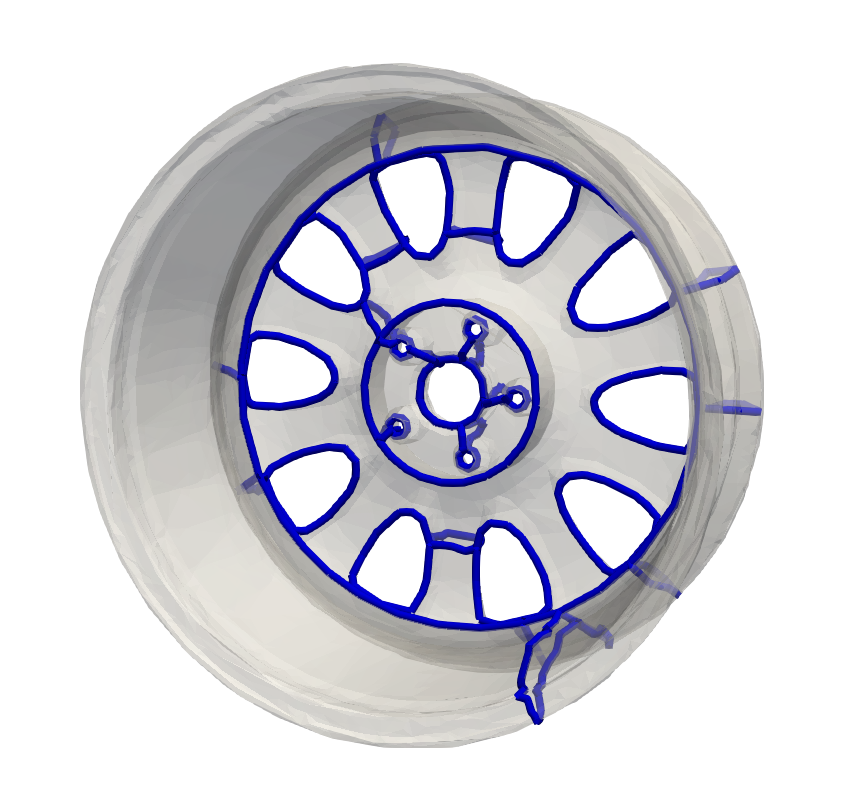}&
\includegraphics[height=1.55in]{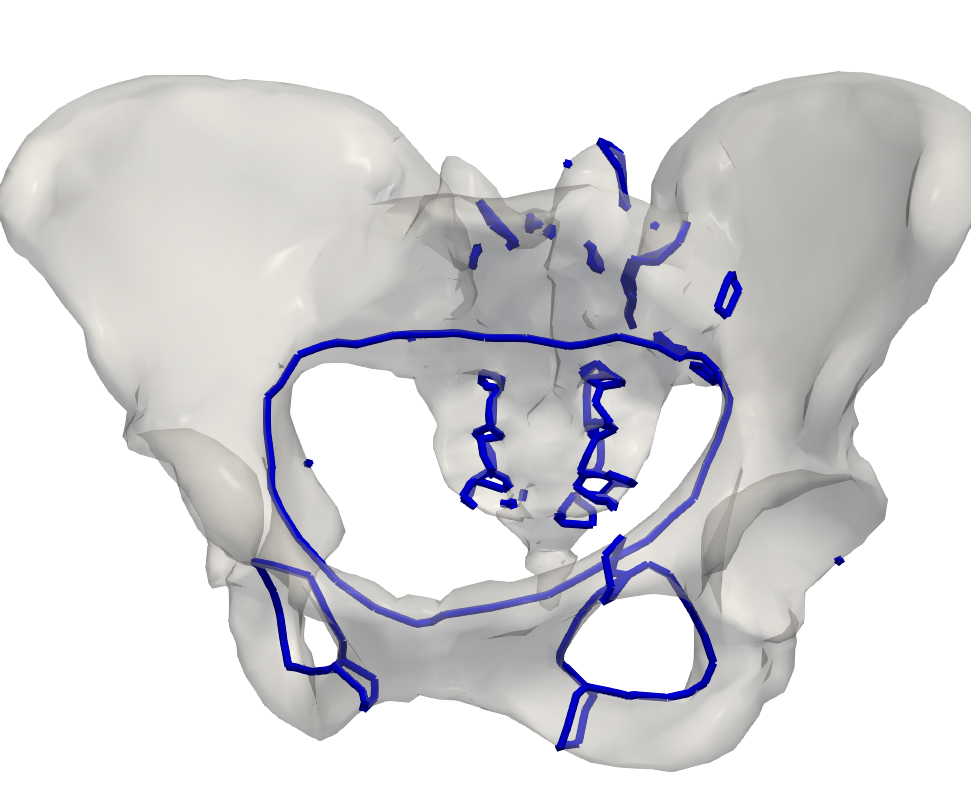}&
\includegraphics[height=1.55in]{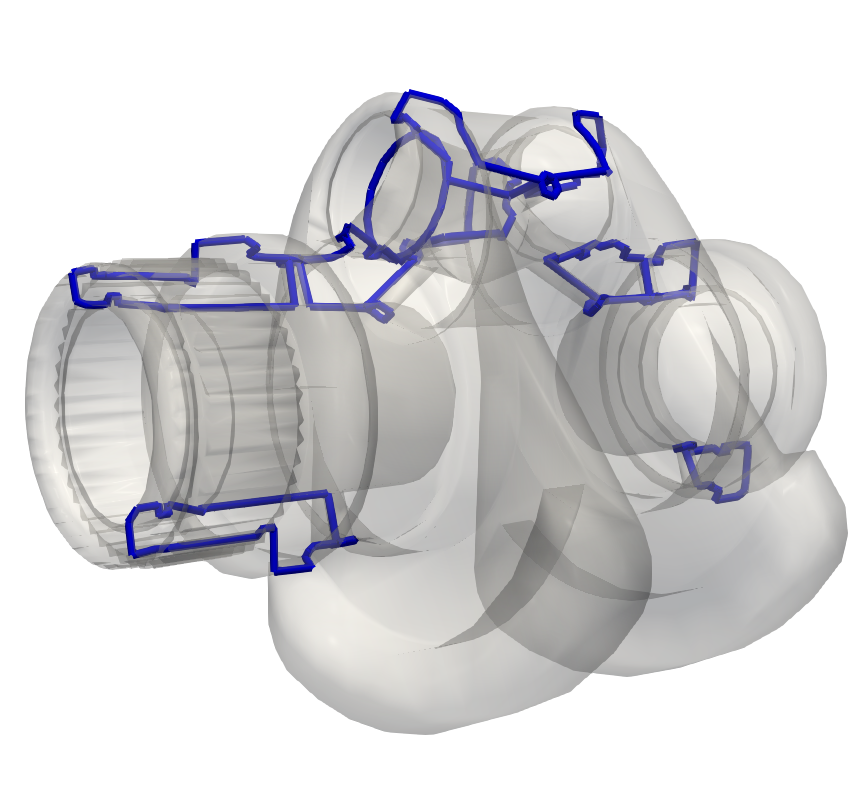} \\
Wheel-1($\beta_1$=36) & Hip Joint($\beta_1$=46) & Crank($\beta_1$=18) \\

\includegraphics[height=1.55in]{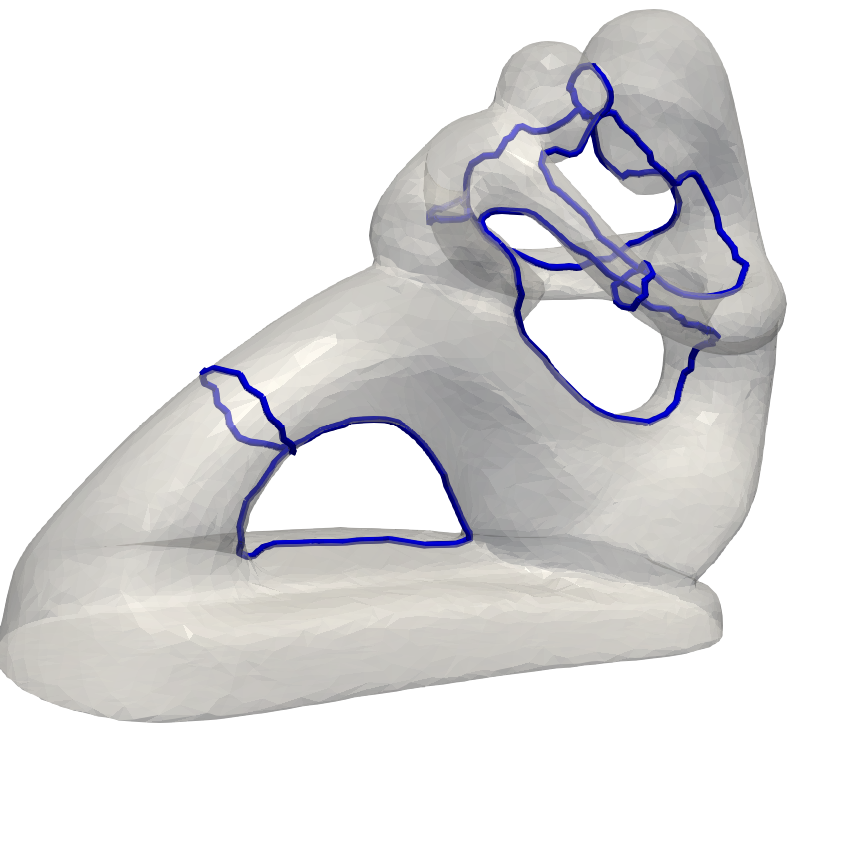} &
\includegraphics[height=1.6in]{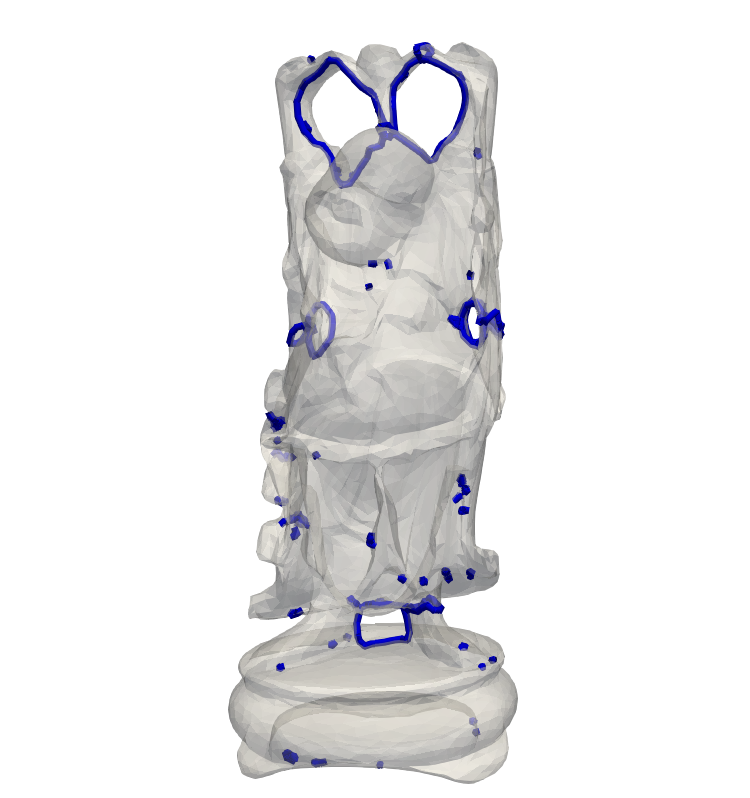} &
\includegraphics[height=1.6in]{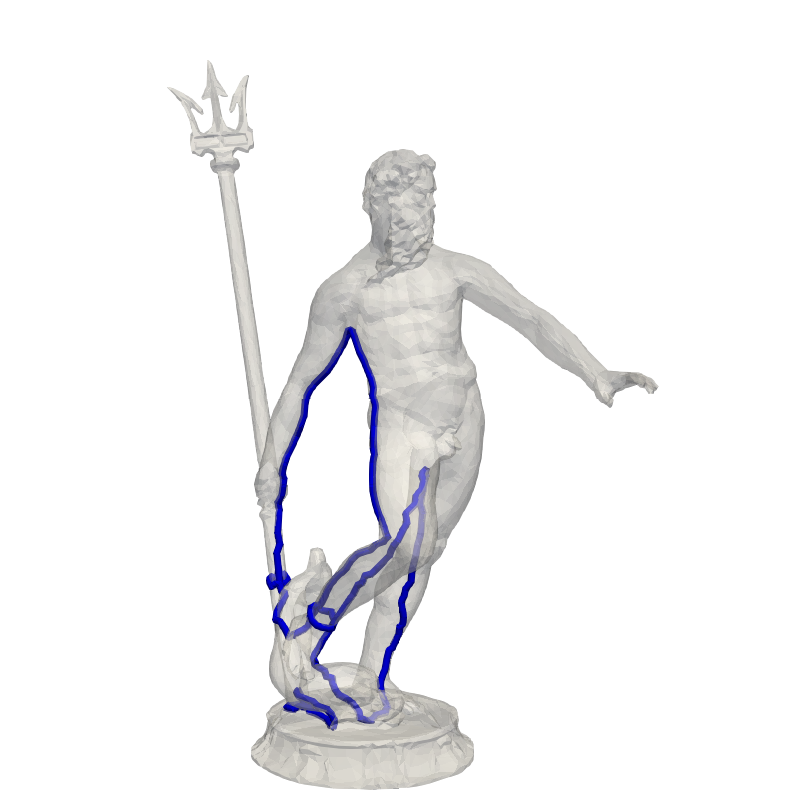} \\

Mother-Child($\beta_1$=8) & Laughing Buddha($\beta_1$=208) & Neptune($\beta_1$=6) \\

\includegraphics[height=1.55in]{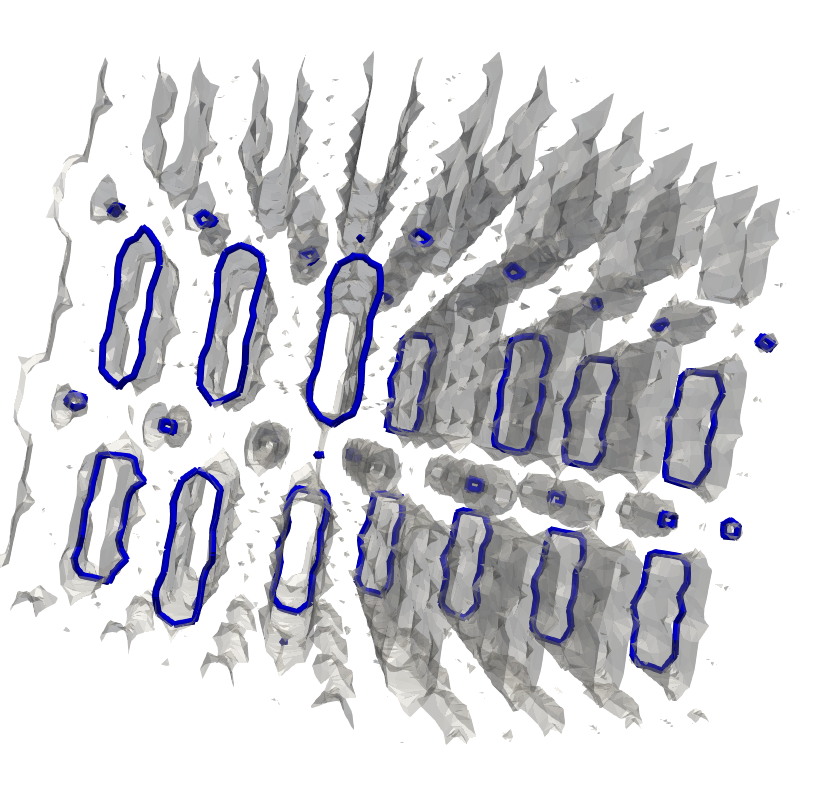} &
\includegraphics[height=1.55in]{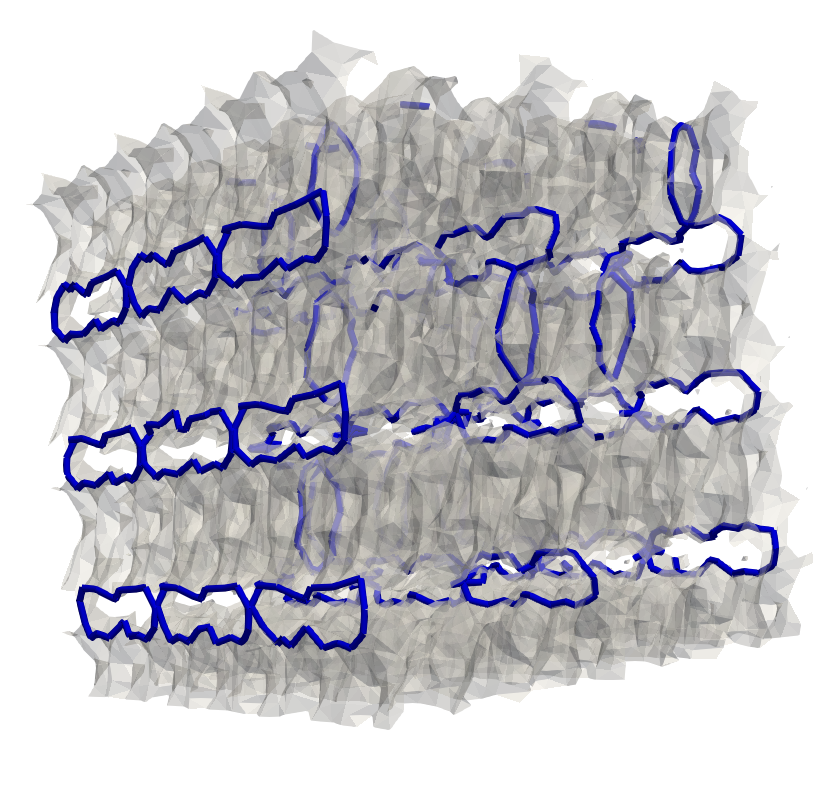} &
\includegraphics[height=1.55in]{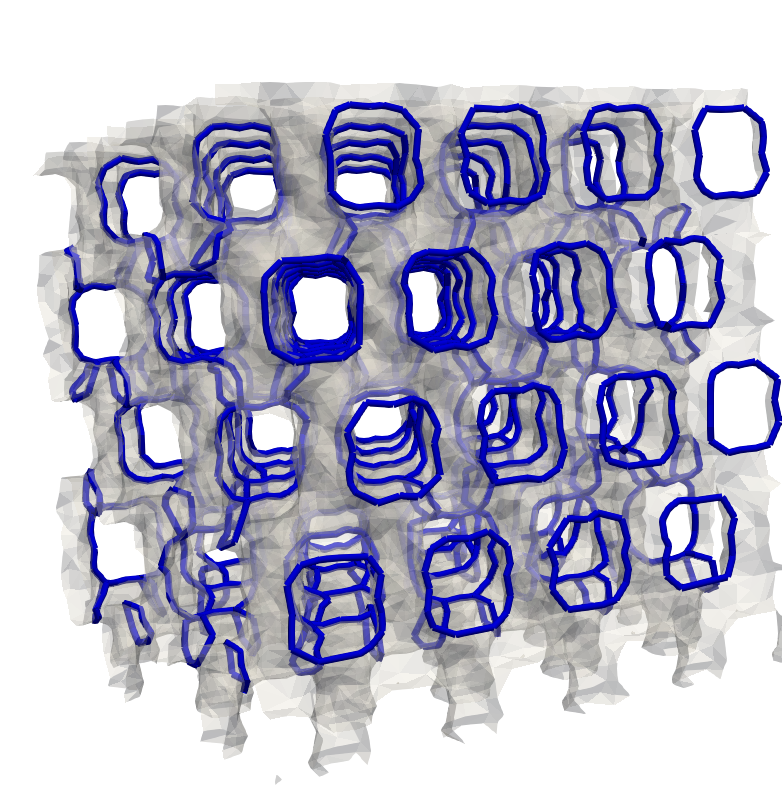} \\
Zeolite-1($\beta_1$=54) & Zeolite-2($\beta_1$=96) & Zeolite-3($\beta_1$=108) \\
\includegraphics[height=1.55in]{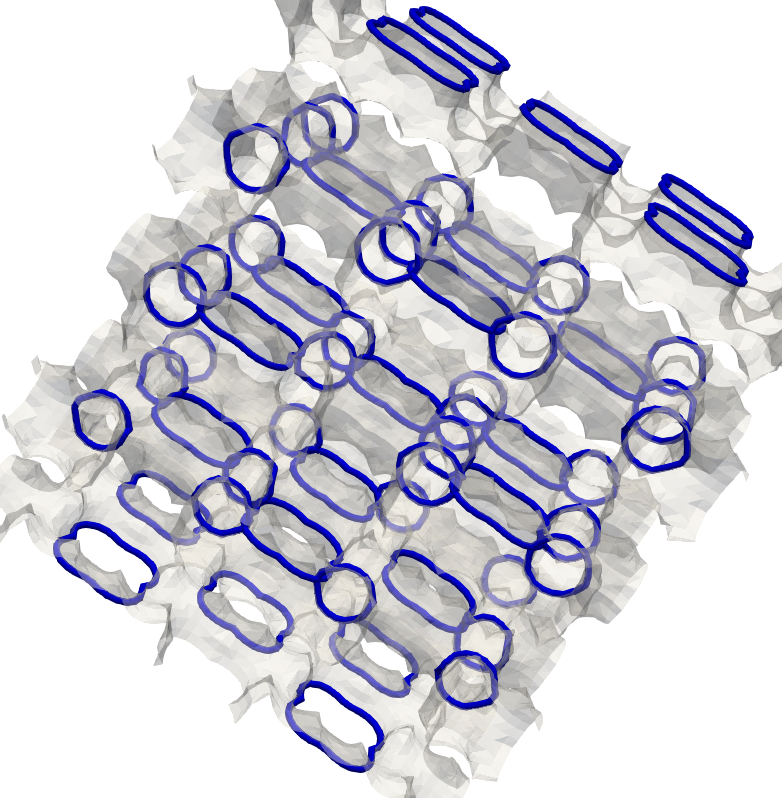} &
\includegraphics[height=1.55in]{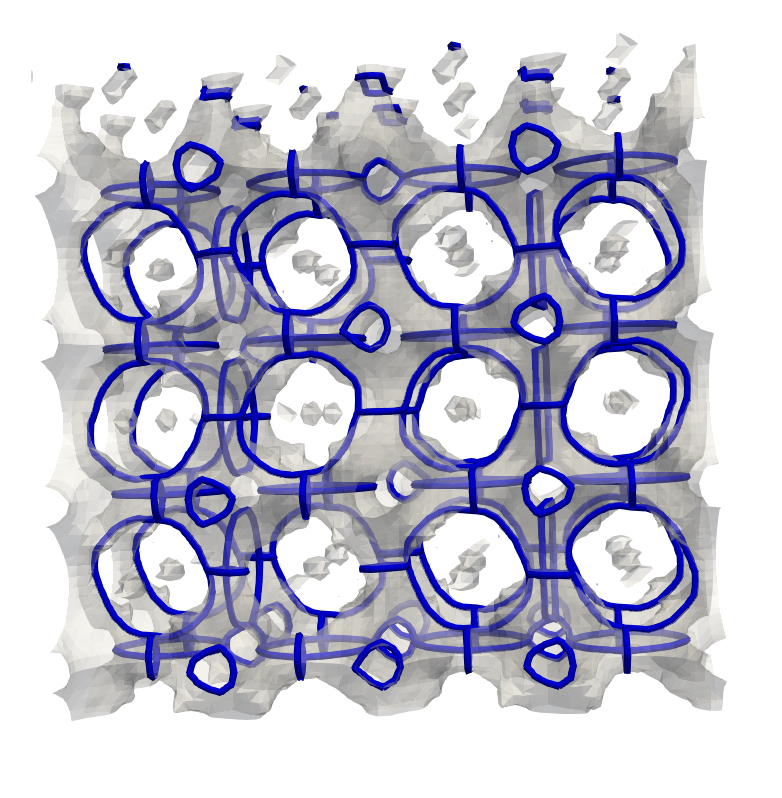} &
\includegraphics[height=1.55in]{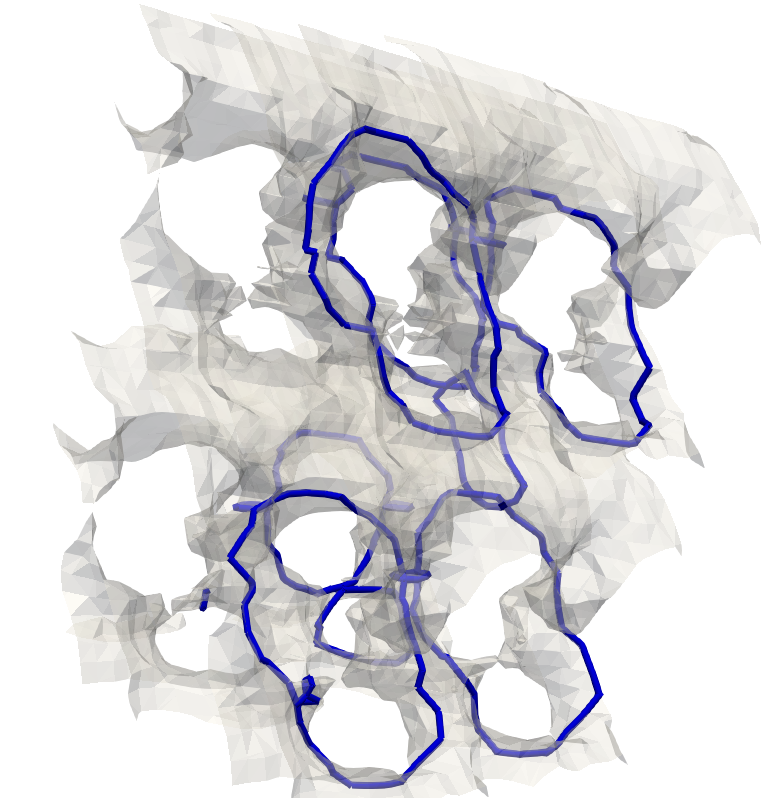} \\
Zeolite-4($\beta_1$=124) & Zeolite-5($\beta_1$=112) & Zeolite-6($\beta_1$=8) \\
\end{tabular}
\caption{Minimum 1-homology basis (blue) computed on various 2D and 3D datasets.}
    \label{fig:Mhb_examples}
\end{figure*}

\begin{figure*}[!htb]
\centering
\begin{tabular}{ccc}
\includegraphics[height=1.6in]{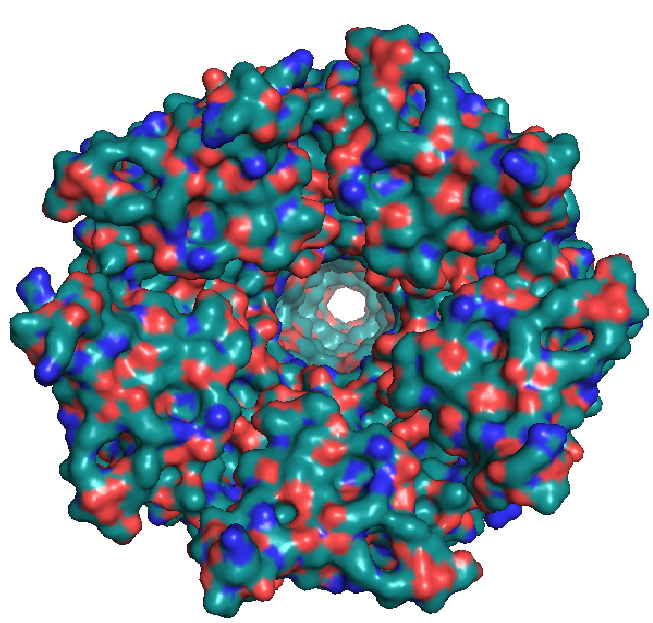}&
\includegraphics[height=1.6in]{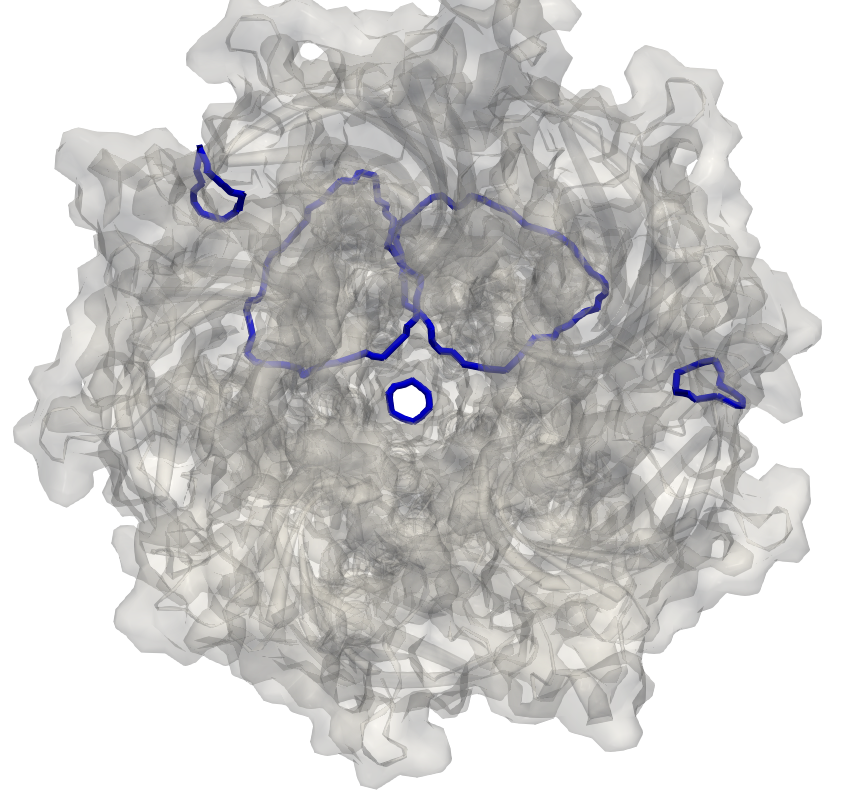}&
\includegraphics[height=1.6in]{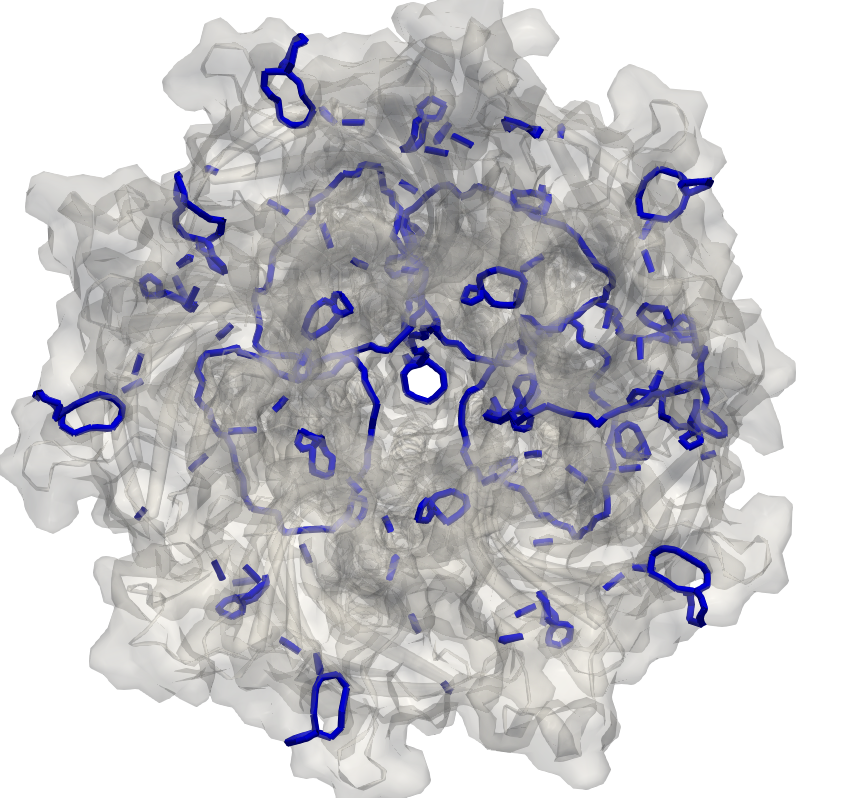}\\
3EAM Molecule & Subset of MHB with central cavity & MHB\\
\includegraphics[height=1.6in]{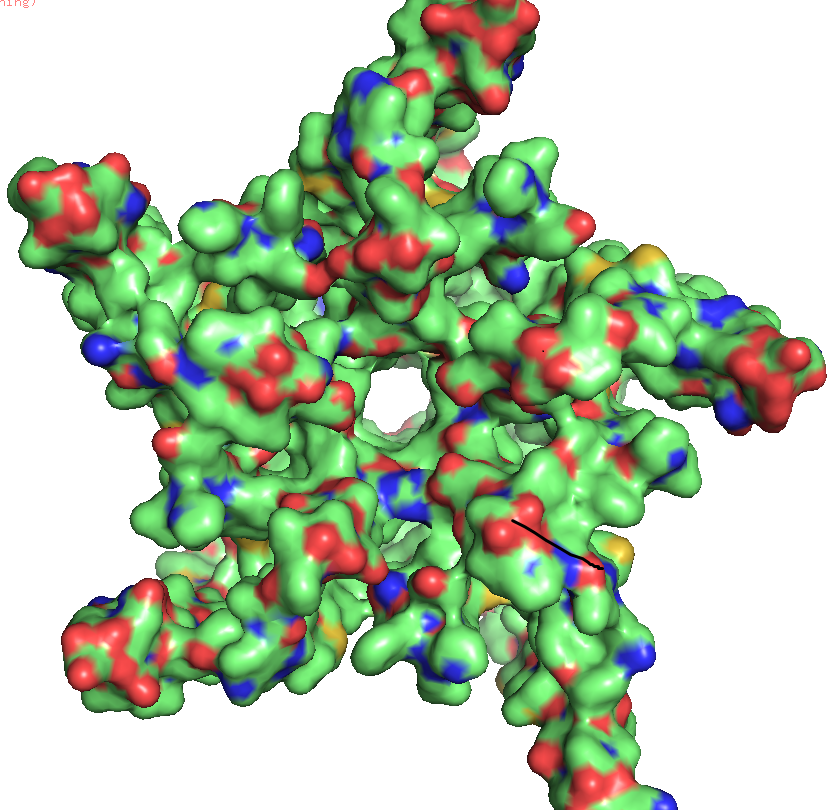} &
\includegraphics[height=1.6in]{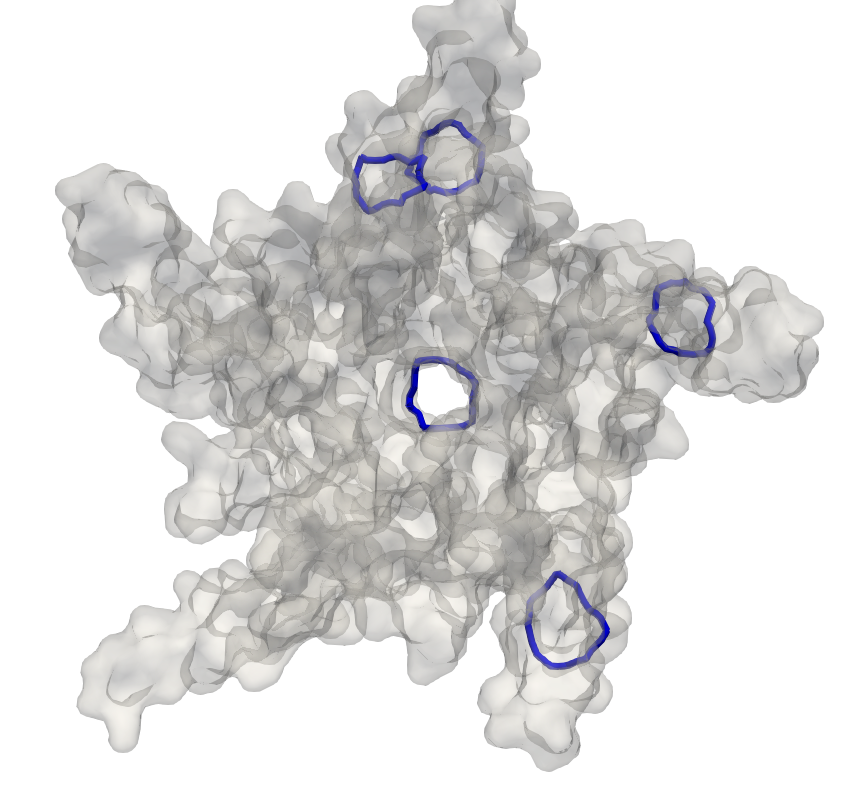} &
\includegraphics[height=1.6in]{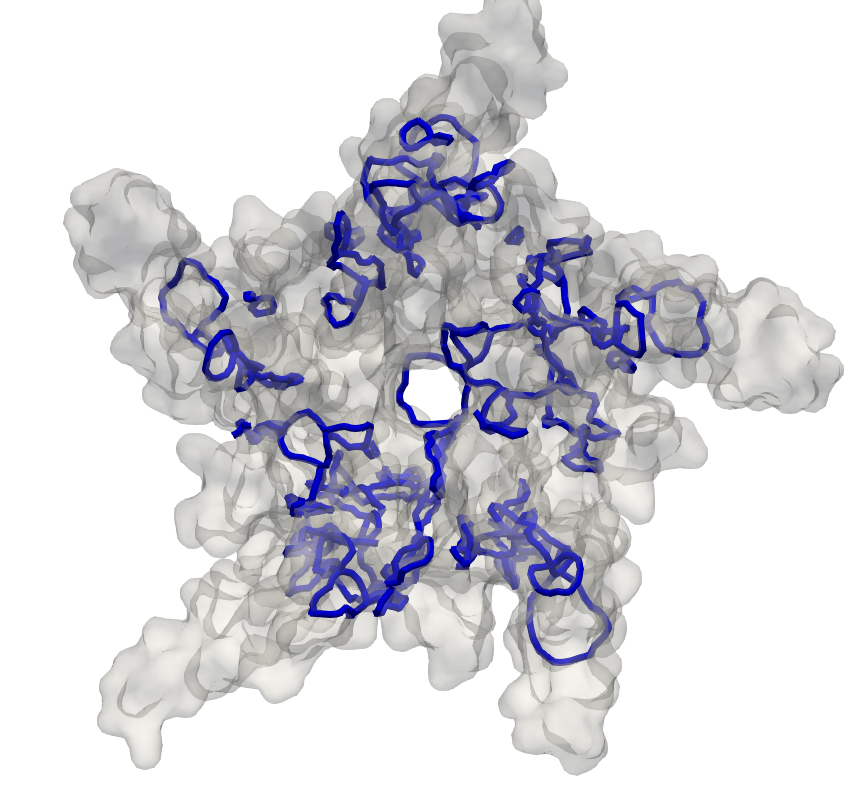} \\
10ED Molecule & Subset of MHB with central cavity & MHB\\
\end{tabular}
\caption{Minimum homology basis computed for two membrane proteins. Left: A rendering of the molecular surface. Right: The minimum homology basis. Middle: A subset of cycles in the basis that represent the central tunnel and some of the larger pores in the protein.}
    \label{fig:Mhb_Protein}
\end{figure*}

\subsection{Synthetic data}
We also present results of experiments on synthetic datasets, consisting of two classes of random complexes, with the aim of studying the scaling behavior of our algorithm. The first class is random clique complexes, a well-studied class of random complexes. Let $n$ be the number of vertices in the complex and let $p$ be a probability parameter. A two dimensional clique complex $C(n,p)$ is constructed as follows: an edge $\{i,j\}$ is chosen in $C(n,p)$ with probability $p$, and a $2$-simplex is included in $C(n,p)$ if all of its edges belong to $C(n,p)$. 

A second class of 2-complexes are random triangle complexes $R(n,p)$, which are constructed as follows.
 Let $n$ be the number of vertices and $p$ be a probability parameter. A triangle $\{i,j,k\}$ is included in $R(n,p)$ with probability $p$. When the triangle $\{i,j,k\}$ is included, all edges on its boundary are also included in $R(n,p)$. 


\subsection{Runtime and scaling}
\label{sec:runtime_comparison}
\Cref{table: execution_time} reports runtimes for both real world and synthetic data. We note that, in general, the runtimes are short for medium sized datasets. The table also presents a comparison against ShortLoop. We observe a significant improvement in execution time over ShortLoop, with a speedup of 1-2 orders of magnitude in some cases. 
ShortLoop does not terminate within a reasonable time frame (~1 day) for the synthetic data. So, we omit a comparison on these random complexes.

 The synthetic datasets help study the scaling behavior of \fastloop. We choose three probability parameters for each class of random complexes, ($0.025, 0.05, 0.075$) for random clique complexes and ($5 \times 10^-5, 10^{-4}, 2 \times 10^{-4}$) for random triangle complexes. We fix value of the probability parameter $p$ and study the runtime performance against increasing size of the complex. 
\Cref{plot:total_runtime} shows a log-log plot of the total running time against number of simplices in the input complex. \Cref{plot:total_runtime_edges} shows a log-log plot of the total  running time against the number of edges in the complex. Both plots are linear indicating that the runtime is a power function of the form $ax^b$. We fit a straight line to the points on the plot and compute its slope in order to estimate the exponent $b$ of the power function. 

The value of the exponent for the plot of runtime vs. total simplices in \Cref{plot:total_runtime} is at most $1.8$ for clique complexes and at most $2$ for random triangles complexes. The exponent is at most $2$ for both classes of complexes in the plot of runtime vs. number of edges.
The empirically observed complexity is better than the theoretical worst case runtime complexity of \Cref{alg:hombasistwo}.
\Cref{MCB-size-plot} shows a plot of MCB size, the cardinality of the multi-set of edges in the minimum cycle basis. The MCB is the output of the first step of the algorithm. The scaling behavior of the MCB size is indicative of the rate determining step of the algorithm. Indeed, a comparison between \Cref{MCB-len-time-plot_Rt} and \Cref{plot:total_runtime} shows that the time taken for computing the MCB is several orders of magnitude greater than the subsequent steps, and that it constitutes a large fraction of the overall runtime.

%
\begin{table}[!h]
    \centering
    \begin{tabular}{lrrrrr}
    \toprule
    dataset & \#Vertices & \#Edges & \#Triangles & Time & Time\\
     &  & & & (\fastloop) & (ShortLoop)\\\midrule
    Wheel-1 & 6970 & 21000 & 14000 & 35s & 245s \\
    Wheel-2 & 2476 & 7500 & 5000 & 7s & 80s \\
    Genus3 & 2462 & 7500 & 5000 & 7s & 84s \\
    Mother-Child & 6494 & 19500 & 13000 & 34s & 40s \\
    Neptune & 6246 & 18750 & 12500 & 30s & 30s \\
    Zeolite-1 & 17275 & 46084 & 30000 & 125s & 80s \\
    Zeolite-2 & 13410 & 38200 & 25000 & 90s & $>$ 2 hr \\
    Zeolite-3 & 12962 & 38116 & 24999 & 100s & $>$ 2 hr \\
    Zeolite-4 & 7629 & 21741 & 14128 & 300s &  $>$ 2 hr \\
    Zeolite-5 & 24643 & 72075 & 47435 & 360s &  $>$ 2 hr \\
    Zeolite-6 & 24573 & 71587 & 46963 & 300s & $>$ 2 hr\\
    Protein-1(3EAM) & 30374 & 90163 & 60000 & 480s & $>$ 2 hr\\
    Protein-2(1OED) & 29846 & 89949 & 59999 & 540s & $>$ 2 hr\\
    \botrule
    \end{tabular}
    \caption{Runtime analysis. \fastloop outperforms ShortLoop in terms of total runtime with a speedup of 10$\times$ or more in many instances.}
    \label{table: execution_time}
\end{table}


\begin{figure}[!h]
\centering
\includegraphics[width=0.8\linewidth]{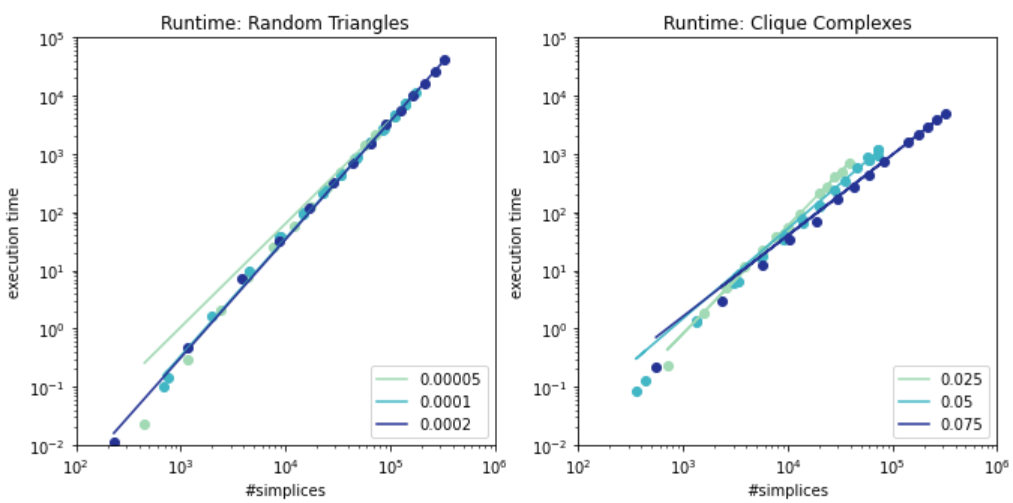}
\caption{Scaling study. A log-log plot of running time vs. total number of simplices in the input complex. The running time scales at most quadratically with the number of simplices.}
    \label{plot:total_runtime}
\end{figure}
\begin{figure}[!h]
\centering
\includegraphics[width=0.8\linewidth]{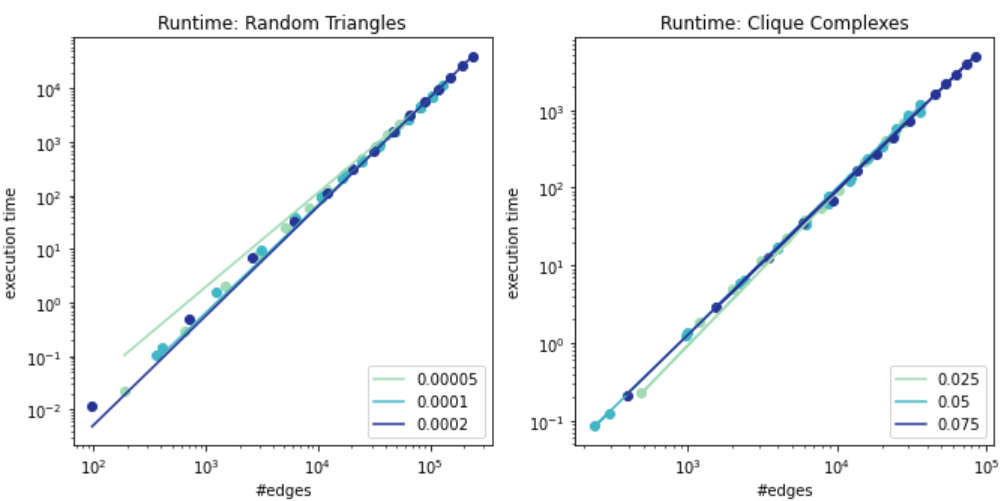}
\caption{Scaling study. A log-log plot of running time vs. number of edges in the input complex. The running time scales quadratically.}
    \label{plot:total_runtime_edges}
\end{figure}
\begin{figure}[!h]
\centering
\includegraphics[width=0.8\linewidth]{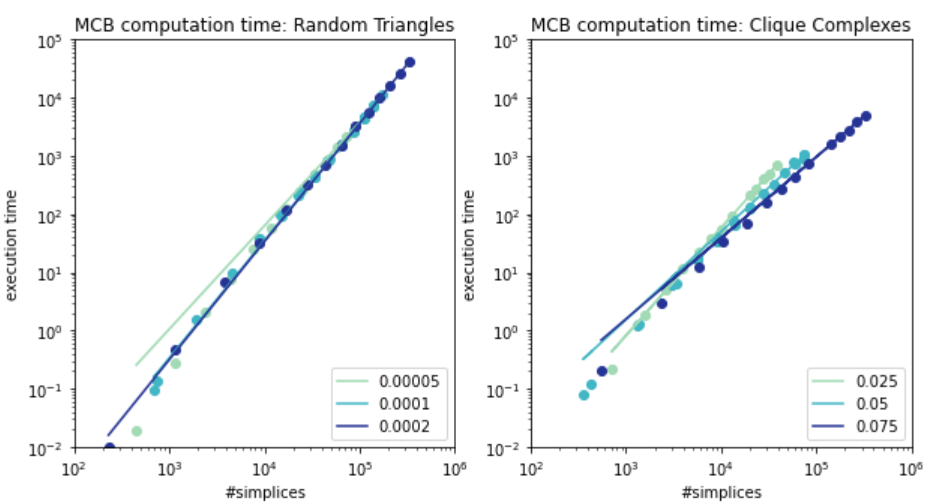}
\caption{A log-log plot of the time taken to compute the MCB vs. number of simplices in the input indicates that this step scales quadratically with the input.}
    \label{MCB-len-time-plot_Rt}
\end{figure}
\begin{figure}[!h]
\centering
\includegraphics[width=0.8\linewidth]{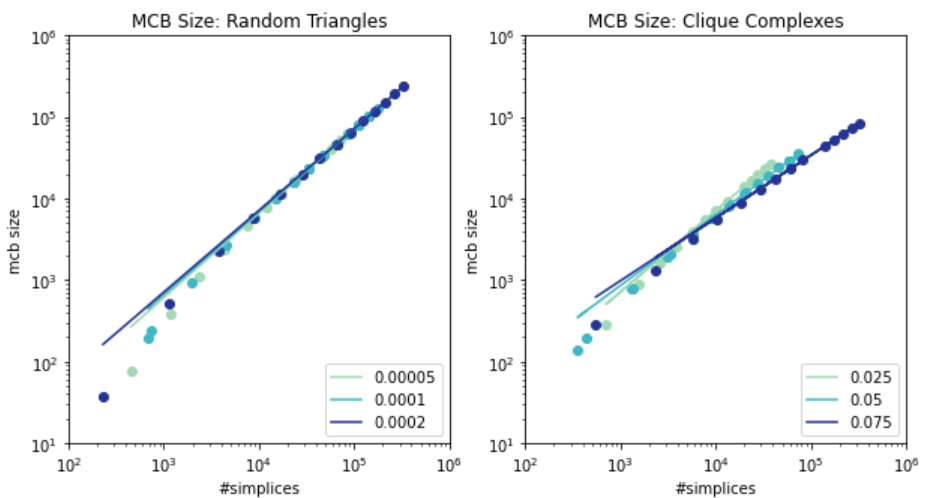}
\caption{The number of edges in the MCB increases quadratically with the size of the input, which explains why computing the MCB is the rate determining step.}
    \label{MCB-size-plot}
\end{figure}





\section{Discussion}
In this paper, we show that questions about minimum cycle basis and minimum homology basis can be naturally recast into the problem of computing rank profiles of matrices, leading to fast algorithms with simple and elegant high-level descriptions. 
The column rank profile (or the earliest basis) of a matrix has previously been used to compute the minimum homology basis of a simplicial complex~\cite{Busaryev,DeyLatest}. Such a greedy approach that picks, at each step, an independent cycle of the smallest index, works because of the matroid structure of homology bases and cycle bases. The novelty of our approach is the observation that independence can be efficiently checked owing to the sparsity of the matrices comprising of candidate cycles.

\changed{In \Cref{sec:outsense}, we describe a randomized $g$-sensitive algorithm for computing minimum homology basis that runs in nearly quadratic time when $g = O(1)$. We believe this is the first such algorithm for this problem for general complexes. }

Experiments on real-world data sets reveal how \fastloop captures the one dimensional \enquote{holes} that may be useful in a variety of practical applications.
\fastloop computes the minimum homology basis of a variety of medium to large sized real-world data sets within a few minutes and consistently outperforms  the state of the art implementation, ShortLoop.
The algorithm as well as the software consist of two major components, namely computing \changed{ a }minimum cycle basis followed by a reduction step. The two components are based on independent algorithms, which may be replaced \changed{in the future} with alternate methods to achieve better theoretical complexity or practical running times.

\backmatter

\bmhead{Supplementary information}
The accompanying video (Online Resource~1) shows the computed cycle representatives from different view points.

\bmhead{Acknowledgements}
This work is partially supported by the PMRF, MoE Govt. of India, an NSF grant CCF 2049010, and a SERB grant CRG/2021/005278. 
VN acknowledges support from the Alexander von Humboldt Foundation, and Berlin MATH+ under the Visiting Scholar program. Part of this work was completed when VN was a guest Professor at the Zuse Institute Berlin.

\bibliography{min-hom}

\end{document}